\documentclass[letterpaper,twocolumn,10pt]{article}
\usepackage{formatting/usenix-2020-09}

\usepackage[utf8]{inputenc} %
\usepackage[T1]{fontenc}    %
\usepackage{hyperref}       %
\usepackage{url}            %
\usepackage{booktabs}       %
\usepackage{amsfonts}       %
\usepackage{nicefrac}       %
\usepackage{microtype}      %
\usepackage{algorithm2e}
\usepackage{amstext,amssymb,amsmath}
\usepackage{amsthm}
\usepackage{xcolor}
\usepackage{bm}
\usepackage{cleveref}
\usepackage{todo}
\usepackage{graphicx}
\usepackage{caption}
\usepackage{subcaption}
\usepackage{xspace}
\usepackage{balance}
\usepackage{paralist}
\usepackage[numbers,sort&compress]{natbib}
\usepackage[framemethod=TikZ]{mdframed}
\usepackage{tabularx}
\usepackage{siunitx}
\usepackage{thm-restate}

\newcommand{\codom}{\mathrm{codom}}
\newcommand{\dataspace}{\mathcal{Z}}

\newcommand{\btheta}{{\mathbb \theta}}
\newcommand{\calL}{\ensuremath{\mathcal L}}

\newcommand{\g}{\ensuremath{\mathbf g}}
\newcommand{\Id}{\ensuremath{\mathbf I}}

\newcommand{\TotalVar}{\mathrm{tv}}
\newcommand{\advantage}{\mathrm{Adv}}
\newcommand{\aiattributespace}{\mathcal{A}\xspace}

\newcommand{\train}{\ensuremath{\mathcal{T}}\xspace}
\newcommand{\trainset}{\ensuremath{D}\xspace}

\newcommand{\secretspace}{\ensuremath{\mathbb{S}}\xspace}
\newcommand{\observablespace}{\ensuremath{\mathbb{O}}\xspace}

\newcommand{\challengezero}{\ensuremath{z^*_0}\xspace}
\newcommand{\challengeone}{\ensuremath{z^*_1}\xspace}
\newcommand{\challengeb}{\ensuremath{z^*_b}\xspace}
\newcommand{\challenge}{\ensuremath{z^*}\xspace}
\newcommand{\adv}{\ensuremath{\mathrm{Attacker}}\xspace}

\newcommand{\N}[2]{\mathcal{N}\!\left(#1, #2\right)}
\newcommand{\ber}[1]{\ensuremath{\mathcal{B}}\left(#1\right)}

\newcommand{\concat}[2]{\ensuremath{#1 \mid #2}}
\newcommand{\prop}{\ensuremath{f}}
\newcommand{\samplerate}{\ensuremath{p}\xspace}
\newcommand{\tpr}{\ensuremath{\mathrm{TPR}}\xspace}
\newcommand{\fpr}{\ensuremath{\mathrm{FPR}}\xspace}

\newcommand{\tp}{\ensuremath{\mathrm{TP}}\xspace}
\newcommand{\fp}{\ensuremath{\mathrm{FP}}\xspace}

\newcommand{\fn}{\ensuremath{\mathrm{FN}}\xspace}
\newcommand{\positives}{\ensuremath{\mathrm{P}}\xspace}
\newcommand{\negatives}{\ensuremath{\mathrm{N}}\xspace}

\newcommand{\adult}{\texttt{Adult}\xspace}
\newcommand{\purchase}{\texttt{Purchase}\xspace}

\newcommand{\fnorm}[1]{\ensuremath{\|#1\|_{F}}\xspace}

\newcommand{\eg}{e.g.\xspace}
\newcommand{\ie}{i.e.\xspace}

\newtheorem{theorem}{Theorem}
\newtheorem{corollary}[theorem]{Corollary}

\newtheorem{proposition}[theorem]{Proposition}

\RestyleAlgo{ruled}
\newenvironment{game}[1][htb]
  {
   \begin{algorithm}[#1]%
}{\end{algorithm}}

\SetCommentSty{mycommfont}

\DeclareMathOperator{\erf}{erf}

\definecolor{OliveGreen}{rgb}{0,0.6,0}

\newcommand{\newcode}[1]{\textcolor{OliveGreen}{#1}} %

\begin{document}

\date{}

\title{\Large \bf Closed-Form Bounds for DP-SGD against Record-level Inference}

\author{
{\rm Giovanni Cherubin\thanks{Corresponding author.}}\\
Microsoft Security Response Center\\
\and
{\rm Boris Köpf}\\
Microsoft Azure Research\\
\and
{\rm Andrew Paverd}\\
Microsoft Security Response Center\\
\and
{\rm Shruti Tople}\\
Microsoft Azure Research\\
\and
{\rm Lukas Wutschitz}\\
Microsoft M365 Research\\
\and
{\rm Santiago Zanella-Béguelin}\\
Microsoft Azure Research
}

\maketitle

\begin{abstract}
Machine learning models trained with differentially-private (DP) algorithms such as DP-SGD enjoy resilience against a wide range of privacy attacks.
Although it is possible to derive bounds for some attacks based solely on an $(\varepsilon,\delta)$-DP guarantee, meaningful bounds require a small enough privacy budget (\ie, injecting a large amount of noise), which results in a large loss in utility.
This paper presents a new approach to evaluate the privacy of machine learning models against specific record-level threats, such as membership and attribute inference, without the indirection through DP.
We focus on the popular DP-SGD algorithm, and derive simple closed-form bounds.
Our proofs model DP-SGD as an information theoretic channel whose inputs are the secrets that an attacker wants to infer (\eg, membership of a data record) and whose outputs are the intermediate model parameters produced by iterative optimization.
We obtain bounds for membership inference that match state-of-the-art techniques, whilst being orders of magnitude faster to compute.
Additionally, we present a novel data-dependent bound against attribute inference.
Our results provide a direct, interpretable, and practical way to evaluate the privacy of trained models against specific inference threats without sacrificing utility.
\end{abstract}

\section{Introduction}
\label{sec:introduction}

Privacy of training data is a central concern when deploying Machine Learning (ML) models.
Privacy risks encompass a variety of adversary goals with corresponding threat models.
For example, if one wanted to prevent an attacker with access to a model from inferring whether a specific data record was in the training data, we would aim to train the model to make it resilient against \emph{membership inference} attacks~\cite{Yeom:2020,Li:2013,Shokri:2017}.
On the other hand, if the concern is an attacker uncovering sensitive attributes about training data records, we would ensure resilience against \emph{attribute inference} attacks~\cite{Yeom:2020,Fredrikson:2015}.

\begin{figure}
	\centering
	\includegraphics[width=\linewidth]{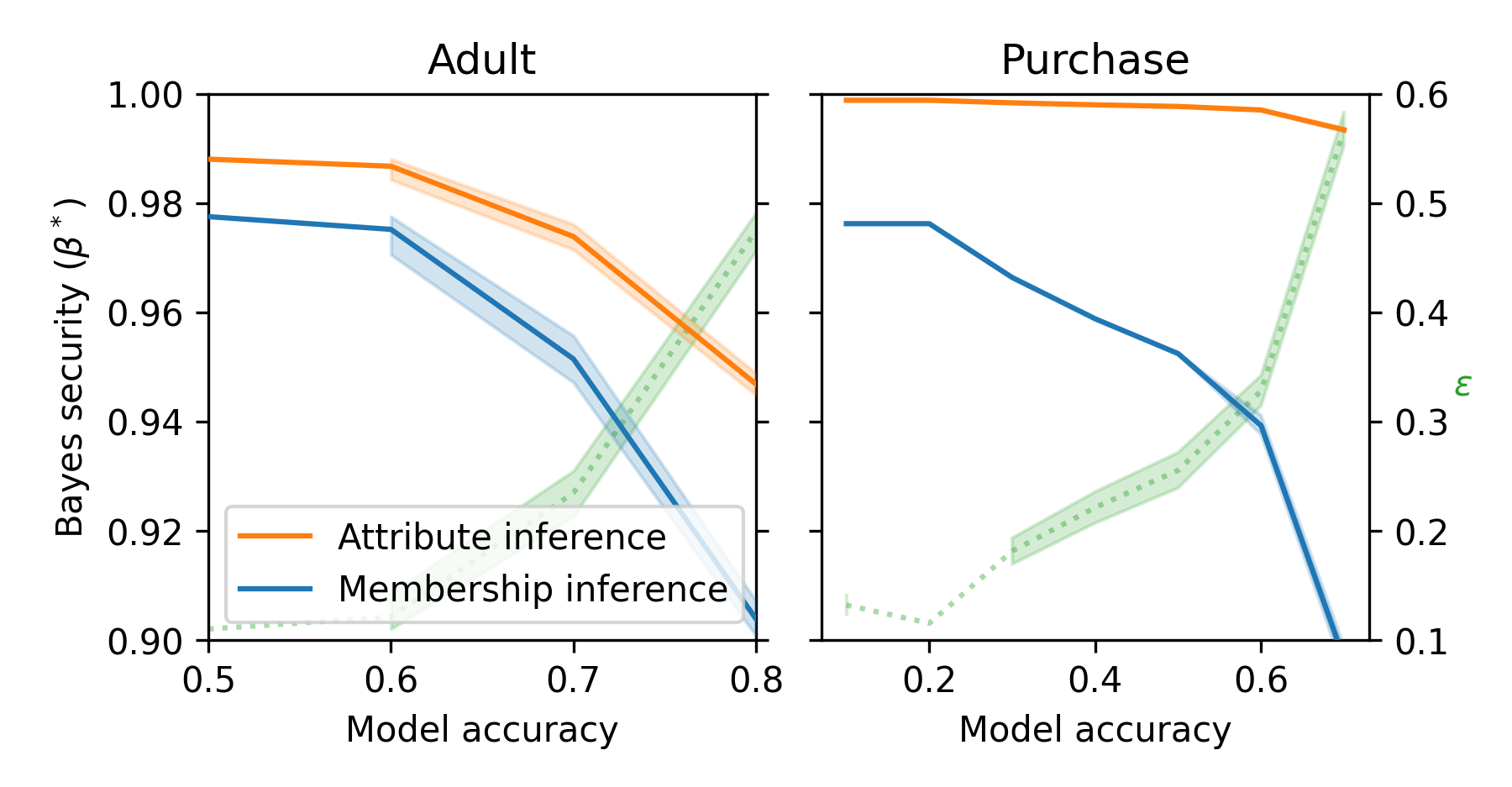}
	\caption{Bayes Security ($\beta^*$) of DP-SGD against MIA and AI on the \adult and \purchase datasets, w.r.t. the accuracy of the model; a higher $\beta^*$ means a more secure model.
	When possible, picking the weaker AI threat model enables achieving a better privacy-utility trade-off.
	For reference, we report the corresponding $(\varepsilon, \delta)$-DP (dashed green line), for $\delta=3.8\times 10^{-6}$ (\adult) and $\delta=4\times 10^{-7}$ (\purchase).
	}
	\label{fig:adult-purchase}
\end{figure}

In practice one may be mostly concerned about some \textit{specific} privacy risks, such as membership or attribute inference.
However, because practitioners lack tools to analyze and mitigate these specific risks, they resort to enforce Differential Privacy (DP), which regards \emph{any} leakage of information about individual records as a privacy violation.
From a theoretical perspective, this choice is convenient: with suitable parameters $(\varepsilon, \delta)$, DP provides quantifiable resilience against all threats to individual training data records.
There are numerous ways of numerically accounting for the privacy budget $(\varepsilon, \delta)$ spent when training a model~\cite{gopi2021numerical,ctd,fft1}, but few ways of computing bounds against specific privacy attacks~\cite{mahloujifar2022optimal}.

\paragraph{Threat-agnostic.}
Firstly, the definition of $(\varepsilon, \delta)$-DP is generally applied in a threat-agnostic manner.
However, in practice there are cases where specific threats give rise to privacy concerns while others do not.
For example, the fact that a person participated in the Census dataset is not privacy sensitive; but if an attacker were able to infer the values of sensitive attributes such as race or age, we would rightly regard this as a privacy violation.
Furthermore, there is no principled way to choose \emph{interpretable} values for $(\varepsilon, \delta)$ without considering a specific privacy threat.
Even when the threat is specified, we still need to find a relationship between this threat and $(\varepsilon, \delta)$ in order to evaluate the risk;
for example, prior work has explored the relationship between DP and membership inference~\cite{Humphries,ZWTSRPNKJ22,chatzikokolakis2020bayes}.
This raises the question: if our aim is to protect against specific threats, can we evaluate our models directly against these threats?

\paragraph{Implementation challenges.}
Secondly, it is known that implementations of $(\varepsilon, \delta)$-DP accountants can be error-prone.
This may be due to implementation difficulties~\cite{fft2,doroshenko2022connect} or numerical errors (\eg, floating point precision)~\cite{gopi2021numerical}.
Further, despite being considered optimal (up to discretization error), accountants generally come with computational costs, which researchers are currently trying to reduce~\cite{doroshenko2022connect}.

\paragraph{Our approach.}
In this paper, we show that it is possible to directly evaluate a trained model against specific privacy threats, such as membership and attribute inference, without actually performing these (often computationally expensive) attacks.
We focus on the mainstream training algorithm DP-SGD, and derive simple closed-form bounds against these threats.
At the core of our proof technique is the approximation of the distribution of intermediate gradients produced by DP-SGD with a Gaussian distribution.
We characterize the approximation error, and show that it can be made negligible by tuning the privacy parameters of the algorithm;
importantly, the error gets smaller for parameters that ensure good privacy.

Our theoretical analyses are facilitated by use of the \emph{Bayes security} metric ($\beta^*$)~\cite{chatzikokolakis2020bayes}.
The main benefit of this metric is its interpretability: it corresponds to the complement of the attacker's advantage, which is widely used in the privacy-preserving ML literature (\eg,~\cite{Yeom:2020}).
Furthermore, Bayes security is threat model specific, prior independent, and one can easily match it to the (optimal) attacker's accuracy for a specific prior.
Additionally, we prove that Bayes security bounds the true positive rate (TPR) of an attacker aiming for a certain false positive rate (\ie TPR$@$FPR), which captures particularly well the risk of membership inference~\cite{carlini2022membership}.

Overall, the simplicity of our proofs suggests our techniques can be extended to study other algorithms and privacy threats.
We summarize our contributions as follows:

\begin{itemize}

\item We propose a new approach to directly measure the privacy of ML models trained using DP-SGD, which addresses the drawbacks outlined above:
1) It is \emph{threat-specific}.
2) It streamlines the proof, in that the metric is directly computed without going via $(\varepsilon, \delta)$,
and it makes for a straightforward implementation. Importantly,
it is orders of magnitude faster to compute than state-of-the-art methods for measuring the risk against MIA~\cite{gopi2021numerical,ctd}.\

\item We demonstrate that our new approach matches (tight) existing techniques in computing bounds for membership inference (MIA)
while requiring orders of magnitude lesser computation time than prior work. %

\item We show a relationship between Bayes security and TPR$@$FPR, a standard metric for MIA~\cite{carlini2022membership}.

\item We use our new approach to compute bounds for attribute inference (AI).
From our bounds, we observe that DP-SGD is significantly more secure against AI than MIA.
This is important because, if a practical application requires security against AI but not MIA, one can achieve a better utility whilst maintaining acceptable privacy, as shown in \Cref{fig:adult-purchase}.

\end{itemize}

Our results, as well as those in the previous literature, assume that an \emph{attacker} has access to intermediate model weights during training.
However, a more realistic (\emph{inference-time}) attacker only has access to the final weights of the model.
To assist future research effort, we also report on our unsuccessful attempts towards obtaining tighter bounds for inference-time attackers.
We show how our framework can model this scenario, and discuss what problems one may need to solve in order to obtain such tighter bounds.

\section{Background and Preliminaries}
\label{sec:background}
We study the security of DP-SGD against record-level inference, with a focus on MIA and AI.
In this section, we provide an overview of the DP-SGD algorithm, we formally define MIA and AI, and we describe
the security metric we use to quantify the resilience of an ML model against both threats.

\subsection{DP-SGD}

Proposed by \citet{Abadi:2016}, Differentially Private Stochastic Gradient Descent (DP-SGD) is a modification of SGD to satisfy $(\varepsilon, \delta)$-DP, as shown in Algorithm~\ref{alg:dpsgd}.
Consider a training set of data records $\{z_1,\ldots,z_{N}\} \in \dataspace^N$ and a loss function $\calL(\btheta)=\frac{1}{N}\sum_i \calL(\btheta, z_i)$, based on model weights $\btheta$.
Let $\eta_t$ be a learning rate, $\sigma$ a noise scaling factor, $C$ a gradient clipping norm, and $L/N$ a sampling factor.
DP-SGD trains a model $\btheta_T$ as follows:
for each step $t = 1, ..., T$, sample on average $L$ records from the training set, clip their gradients' norms to $C$, and add Gaussian noise to their sum;
use the resulting \textit{noisy gradient} $\tilde{\g}_t$ to update the model weights according to the learning rate, and repeat for the desired number of steps.

\begin{algorithm}[htb!]
	\caption{$\textrm{DP-SGD}(\{z_1,\ldots,z_{N}\},\ \calL(\btheta),\ \eta_t,\ \sigma,\ L,\ C)$}
	\label{alg:dpsgd}
	\DontPrintSemicolon
	{\bf Initialize} $\btheta_0$ randomly\;
	\For{$t \in [T]$}{
		Take a random sample $L_t$ with sampling probability $L/N$\;
		{\bf Compute gradient}\;
		{For each $i\in L_t$, compute $\g_t(z_i) \gets \nabla_{\btheta_t} \calL(\btheta_t, z_i)$}\;
		{\bf Clip gradient}\;
		{$\bar{\g}_t(z_i) \gets \g_t(z_i) / \max\big(1, \frac{\|\g_t(z_i)\|_2}{C}\big)$}\;
		{\bf Add noise}\;
		{$\tilde{\g}_t \gets \frac{1}{L}\left( \sum_i \bar{\g}_t(z_i) + \mathcal{N}(0, \sigma^2 C^2 \Id)\right)$}\;
		{\bf Descent}\;
		{ $\btheta_{t+1} \gets \btheta_{t} - \eta_t \tilde{\g}_t$}\;
	}
	{\bf Output} $\btheta_T$ %
\end{algorithm}

Typically, the privacy parameters $(\varepsilon, \delta)$ are obtained numerically via \textit{accounting mechanisms};
due to the iterative nature of DP-SGD this is essential to obtain accurate privacy guarantees.
\citet{Abadi:2016} introduced the Moments Accountant for computing the privacy guarantees for composed mechanisms.
More recently, \citet{gdp} introduced $f$-DP which gives rise to lossless composition:
this notion of DP composes all possible $(\varepsilon, \delta)$ at once, and only afterwards it converts the privacy guarantee back to a single $(\varepsilon, \delta)$ pair.
Alas, computing this composition is challenging, and several numerical approximations have been developed \cite{gopi2021numerical,fft1, fft2, ctd}; these are tight up to discretization error.

\subsection{Threat Models}

We consider two specific threat models: membership inference (Game~\ref{game:mia-record-generalized}), and attribute inference (Game~\ref{game:ai}).

\paragraph*{Membership inference (MIA).}
In (record-level) MIA, the attacker aims to ascertain whether a data record appeared in the model's training set.
This threat model is formalized in Game~\ref{game:mia-record-generalized}.
In this game, a challenge point $\challenge_s$ is sampled from a set of challenge points $\{\challenge_i\}_{i=1}^M$ according to an arbitrary prior distribution $\pi$ on this set.
The model is trained on $\trainset \cup \{\challenge_s\}$ for $T$ steps, and the intermediate DP-SGD updates $\{\btheta_{t}\}_{t=1}^T$ are revealed to the attacker; the attacker is also assumed to know the game parameters $\train, T, \pi$, as well as the set of challenge points.
The attacker's goal is to guess \emph{which} of the challenge points was used for training the model.
Game~\ref{game:mia-record-generalized} generalizes common MIA setups in two ways.
First, the number of challenge points $M$ can be larger than $2$.
Second, the game enables associating a prior distribution $\pi$ to the choice of the challenge points.
Thanks to the metric we use (\Cref{sec:bayes-security}),
it will suffice to compute the security against MIA for the two worst-case challenge points ($M=2$) and a uniform prior (\Cref{sec:main-result}).

\begin{game}[hbt!]
	\DontPrintSemicolon
	\caption[F]{$\textrm{MIA-record-level}(\train_T, \trainset, \{\challenge_1, ..., \challenge_M\}, \pi_{\{1, ..., M\}})$}
		$s \gets \pi_{\{1, ..., M\}}$\;
		$\{\btheta_{t}\}_{t=1}^T \gets \train_T(\trainset \cup \{\challenge_s\})$\;
		$s' \gets \adv(\{\btheta_{t}\}_{t=1}^T,\ \{\challenge_i\}_{i=1}^{M}, \trainset, \train_T, \pi_{\{1, ..., M\}})$
	\label{game:mia-record-generalized}
\end{game}

\paragraph*{Attribute inference (AI).}
Let $\challenge = \concat{\varphi}{s}$ be a data record, composed of the concatenation of two vectors: $\varphi$ and $s$.
In attribute inference, the attacker aims to infer the value of one or more \emph{sensitive attributes} of a data record, $s$, given access to the remainder of that record, $\varphi$.
As shown in Game~\ref{game:ai}, the sensitive attribute $s \in \aiattributespace$ is sampled according to some prior
$\pi$ on the set.
The model is trained for $T$ steps on the training set $\trainset \cup \{\challenge\}$, and the intermediate updates $\{\btheta_{t}\}_{t=1}^T$ are revealed to the attacker; the attacker is also assumed to know all the game parameters, including the set of sensitive attributes.
The goal of the attacker is to guess the sensitive attribute $s$.

\begin{game}[hbt!]
	\DontPrintSemicolon
	\caption[F]{$\textrm{AI}(\train_T, \trainset, \varphi, \aiattributespace, \pi_{\aiattributespace})$}
		$s \gets \pi_{\aiattributespace}$\;
		$\challenge \gets \concat{\varphi}{s}$ \tcp*{The attributes are concatenated}
		$\{\btheta_{t}\}_{t=1}^T \gets \train_T(\trainset \cup \{\challenge\})$\;
		$s' \gets \adv(\{\btheta_{t}\}_{t=1}^T,\ \varphi, \aiattributespace, \trainset, \train_T, \pi_{\{1, ..., M\}})$
	\label{game:ai}
\end{game}

\subsection{The Bayes security metric}
\label{sec:bayes-security}

We define a metric of risk for these threats.

\paragraph*{Generalized attacker advantage.}
The commonly-used metric of \emph{advantage} quantifies how much more likely an attacker is to succeed, at either membership or attribute inference, when given access to the trained model, as compared with not having this access.
Formally, suppose the attacker's goal is to guess some secret information, measured by random variable $S$.
Let $\pi$ denote any prior knowledge the attacker has about $S$; mathematically, $\pi$ is a probability distribution on the range of $S$.
We write $\adv(\pi, \theta)$ to indicate an attacker who has access to the model $\theta$ (and with prior knowledge $\pi$), and $\adv(\pi)$ for an attacker who guesses purely based on prior knowledge;
we assume the former is at least as successful as the latter.
The \textit{generalized} advantage~\cite{cherubin2017bayes} is defined to be the difference between the probability of success of these two attackers, normalized to a value between $[0, 1]$: $0$ implies no advantage and $1$ maximal advantage:
\begin{align*}
	\advantage_\pi
	&= \frac{Pr[\adv(\pi, \theta) = S] - Pr[\adv(\pi) = S]}{1-Pr[\adv(\pi) = S]} \,.
\end{align*}

This is a generalized version of the notion of advantage typically used in the literature (\eg, \cite{Yeom:2020}):
by letting $S$ take a binary value and setting $\pi$ to be a uniform distribution over the possible values of $S$, we get $Pr[\adv(\pi) = S] = \nicefrac{1}{2}$;
substituted into the above, this gives the familiar expression for advantage, as used by \citet{Yeom:2020}:
$$\advantage_\pi = 2 Pr[\adv(\pi, \theta) = S] - 1$$

Note that this specific notion of advantage cannot be applied to Games~\ref{game:mia-record-generalized} and~\ref{game:ai} because, in both cases, the secret may have more than two possible values, and they need not come from a uniform prior distribution $\pi$.

\paragraph{Bayes Security.}
Based on the notion of generalized advantage, we use the \emph{Bayes security} metric~\cite{cherubin2017bayes,chatzikokolakis2020bayes}, defined as:
\begin{equation}
	\label{eq:beta}
	\beta^* = 1 - \max_\pi \advantage_\pi \\
\end{equation}
This metric takes values in the range $[0, 1]$, where $1$ indicates perfect security (\ie, no information leakage).
Importantly, the following holds:
\begin{theorem}[{Bayes security and advantage~\cite[Theorem 1]{chatzikokolakis2020bayes}}]
	\label{thm:mia-game-reduction}
	Bayes security is achieved on (equivalently, the generalized advantage is maximized on)
	a uniform prior on two secrets:
	$$\beta^* = 1 - \max_{s_0, s_1 \in \secretspace} \advantage_{u_{s_0, s_1}} \,,$$
	where $u_{s_0, s_1}$ is a uniform prior on two secrets $s_0, s_1 \in \secretspace$, \ie, $Pr[S=s_0] = Pr[S=s_1] = \nicefrac{1}{2}$,
	and $Pr[S=s] = 0$ $\forall s \neq s_0, s_1$.

\end{theorem}

Conveniently, by using Bayes security we further inherit the following relations:
1) Bayes security is directly related to the total variation distance between the two worst-case distributions in the outputs of DP-SGD (\Cref{sec:main-result}), and 2) it is related to $(\varepsilon, \delta)$-DP (\Cref{sec:mia}).

\section{Proof Strategy and Main Result}
\label{sec:main-result}
\begin{table}
\caption{Summary of notation.}
	\begin{tabular}{lp{6cm}}
	\toprule
	Symbol & Meaning\\
	\midrule
	$\challenge$ & Challenge point about which the attacker wishes to learn some property.\\
	$\prop(\challenge)$ & Property of interest.\\
	$O_1, ..., O_T$ & Intermediate weights output by DP-SGD.\\
	$G_1, ..., G_T$ & Intermediate (noisy) gradients output by DP-SGD for the challenge point.\\
	$P_{G(\challenge) \mid S=s}$ & Distribution of the gradient vector $G = (G_1, ..., G_T)$ given a point $\challenge$ s.t. $\prop(\challenge) = s$.\\
	$\sigma$ & Noise parameter.\\
	$C$ & Gradient norm clipping parameter.\\
	$\samplerate = \nicefrac{L}{N}$ & Sampling rate; $N$: training set size, $L$: user-chosen parameter.
\end{tabular}
\end{table}

In this paper, we derive bounds on the security of DP-SGD against record-level MIA (Game~\ref{game:mia-record-generalized}) and AI (Game~\ref{game:ai}).
First, observe that
these threats can be unified in a record-level property inference setup as follows.
Let $\challenge$ be some data record (\ie, the \emph{challenge point}) about which the attacker aims to infer some property $\prop(\challenge)$.
In MIA, the challenge point is chosen from a set of possible challenge points $\challenge \in \{\challenge_1, ..., \challenge_M\}$, and $\prop(\challenge)$ is an index to that set such that $\challenge = \challenge_{\prop(\challenge)}$.
In AI, for an arbitrary challenge point $\challenge$, composed of the concatenation $\challenge = \concat{\varphi}{s}$, where $s \in \mathcal{A}$ represents a sensitive attribute, the property is $\prop(\challenge) = s$.
In our main result, we assume that $\prop$ is a bijection: there is exactly one challenge point $\prop(\challenge)$ for each property $s \in \codom(\prop)$;
observe that this is satisfied by construction in our two threat models.\footnote{We suspect that this assumption can be relaxed: if $\prop$ was not a bijection, there may be two challenge points $\challengezero \neq \challengeone$ satisfying the same property $\prop(\challengezero) = \prop(\challengeone) = s$; but this can only make the attack easier, as it skews the prior distribution on $s$, thereby making it a more probable guess.
}

Let $S = \prop(\challenge)$ be a random variable representing the secret property.
The attacker aims to guess $S$ given the intermediate models output by DP-SGD, which we denote by the random vector $O = (O_0, O_1, ..., O_T)$.
In the spirit of quantitative information flow~\cite{smith2009foundations}, this can be seen as an information theoretic channel, where the relation between $S$ and $O$ is ruled by the posterior distribution $P_{O \mid S}$.
The Bayes security of this channel, $\beta^*(P_{O \mid S})$, measures the \textit{additional leakage} about the secret $S$ that an attacker can exploit by observing $O$.

\paragraph*{Looking ahead.}
In this section, we prove our main result: a bound on the Bayes security of DP-SGD against the record-level property inference attack;
we later specialize this bound to the cases of MIA (\Cref{thm:mia}) and AI (\Cref{thm:ai}).
We proceed as follows:
\begin{enumerate}
	\item First, we show one only needs to measure the risk for the two worst-case property values (challenge points) (\Cref{sec:mia-two-challenge-points}), as a consequence of \Cref{thm:mia-game-reduction}.
	\item We show that the noise on the model weights coming from the training set can be neglected for our analysis: it suffices to compute Bayes security for the gradient of the challenge point $\challenge$. We then observe that gradients follow a Gaussian mixture distribution.
	\item We prove that a mixture of Gaussians can be approximated by a single Gaussian distribution, with an error term that gets smaller as the noise parameter ($\sigma$) increases (\Cref{thm:small-approximation-error}).
	\item We obtain a bound on the Bayes security of DP-SGD by bounding the total variation distance between the distribution of the gradients for the two challenge points
	that correspond to the two worst-case property values.
\end{enumerate}

\subsection{All You Need Is Two Points ...}
\label{sec:mia-two-challenge-points}

In both Games~\ref{game:mia-record-generalized} and~\ref{game:ai}, $S$ takes values from a potentially large set.
Further, its prior distribution may be skewed: some values of $S$ may be more likely than others.
For example, in MIA, one data record may be \emph{more likely} than another to be a member;
this is captured by our formalization in Game~\ref{game:mia-record-generalized},
where $\pi$ may assign more weight to one particular record.

To solve this issue, we apply the fact that the generalized advantage is maximized over a uniform prior; this is an immediate consequence of \Cref{thm:mia-game-reduction} by \citet{chatzikokolakis2020bayes}.
This implies that, when studying the security of DP-SGD against these threats, it is sufficient to limit the range of $S$ to the two values that are the easiest to distinguish for the attacker, and set $\pi$ to the uniform prior.

For example, under the MIA threat model, this means that by measuring the security for just two challenge points ($M=2$) that are equally likely to be members, we obtain a bound on the security for arbitrary values of $M \geq 2$.
Equivalently, for AI, it is sufficient to look at the two \emph{leakiest} attribute values.

\subsection{... and Intermediate Gradients}
\label{sec:just-need-intermediate-gradients}
The second step in our analysis is the observation that an attacker obtains maximal advantage if they are given direct access to the intermediate gradients, rather than model weights.
Formally, the attacker observes a random vector $O(\challenge) = (O_0(\challenge), ..., O_T(\challenge))$, the intermediate model weights;
in our notation, we make the dependence on the challenge point $\challenge$ explicit
where needed.
$O$ is such that $O_0 = \theta_0$ and for $t\geq0$
\[
	O_t(\challenge) = \frac{1}{L}\Bigg( \sum_{i=1}^{N-1}\bar{\g}_t(z_i)\ber{\samplerate} + \bar{\g}_t(\challenge)\ber{\samplerate} + \mathcal{N}(0, \sigma^2 C^2 \Id)\Bigg)
\]
where $\ber{\samplerate}$ a Bernoulli distribution, with $\samplerate = \frac{L}{N}$,
and $\mathcal{N}(0, \sigma^2 C^2 \Id)$ is isotropic Gaussian noise with variance $\sigma^2 C^2$.

We use the notation $P_{O(\challenge) \mid S=s}$ to indicate the distribution of the intermediate weights, conditioned on the fact that the challenge point satisfies $\challenge: \prop(\challenge) = s$.

Now, consider the random vector $G = (G_1, ..., G_T)$:
\begin{equation}
	\label{eq:gradients}
	G_t = \bar{\g}_t(\challenge)\ber{\frac{L}{N}} + \mathcal{N}(0, \sigma^2 C^2 \Id)
\end{equation}

We observe that the distribution $P_{O \mid S}$ can be obtained via postprocessing from $P_{G \mid S}$; since the Bayes security of a channel cannot decrease by postprocessing,
we have:
\begin{corollary}[Consequence of Theorem 4 in \cite{chatzikokolakis2020bayes}]
	\label{thm:gradients-reduction}
	$$\beta^*(P_{O \mid S}) \geq \beta^*(P_{G \mid S}) \,.$$
\end{corollary}

Intuitively, an attacker has a better (or equal) advantage when attacking channel $P_{G\mid S}$ than $P_{O\mid S}$.
The reason is that $G$ carries at least as much information about the challenge point $\challenge$ as $O$.
We shall henceforth study the security of $P_{G\mid S}$.

\paragraph*{Gradients distribution.}
We provide an explicit expression for $P_{G(\challenge)}$, for a generic challenge point $\challenge$.

At each step, the intermediate gradient is a Gaussian, centered either in $\bar{\g}(\challenge)$ with probability $\samplerate$, or $0$ otherwise.
This means that $P_{G}$ is a mixture of $2^T$ Gaussians:
intuitively, $G$ takes values from a Gaussian centered in $(0, ..., 0)$ with probability $(1-\samplerate)^T$ (i.e., the gradient is never sampled), from a Gaussian centered in
$(\bar{\g}(\challenge), 0, ..., 0)$ with probability $\samplerate(1-\samplerate)^{T-1}$ (i.e., $\bar{\g}(\challenge)$ is sampled in the first step only), and so on.
Let $b \in \{0, 1\}^T$ be a binary vector, where $b_t = 1$ means that $G_t$ is centered in $\bar{\g}(\challenge)$, and $b_t = 0$ means that $G_t$ is centered in $0$;
here, the role of $b$ is that of a mask that indicates in which steps the challenge point is sampled.
Then we can write the distribution as:
$$P_{G} = \sum_{b \in \{0, 1\}^T} c_b \prod_{t=1}^T \mathcal{N}(\bar{\g}(\challenge) \odot b, \sigma^2 I) \,,$$
where $x \odot y$ is the Hadamard (i.e., entrywise) product of vectors $x$ and $y$, and $c_b$ is the probability of observing $b$, $c_b = \samplerate^{|b|}(1-\samplerate)^{T-|b|}$, where $|b|$ is the number of $1$'s in $b$.

\subsection{Bayes Security of DP-SGD}

To determine the Bayes security of DP-SGD, we will use the following relation.

\begin{proposition}[Bayes security and total variation~\cite{chatzikokolakis2020bayes}]
	Let $\mathcal{M}: \mathbb{S} \rightarrow \mathbb{O}$
	be a randomized algorithm.
	Then:
	$$\beta^*(\mathcal{M}) = 1-\max_{s_0, s_1 \in \secretspace}\TotalVar(P_{\mathcal{M}(S) \mid S=s_0}, P_{\mathcal{M}(S) \mid S=s_1})\,,$$
	where $P_{\mathcal{M}(S) \mid S}$ is the posterior distribution of the mechanism's output $\mathcal{M}(S)$ given some input random variable $S$, and $\TotalVar$ is the total variation distance\footnote{
		Consider two measures $P$ and $Q$ on the same measurable space $(\dataspace, \mathcal{F})$; their total variation distance is:
		$\TotalVar(P, Q) = \sup_{A \in \mathcal{F}} |P(A) - Q(A)|$.
	}.
	\label{thm:relations-beta-tv}
\end{proposition}

Based on this, we compute the Bayes security of DP-SGD as the maximal total variation between $P_{G\mid S=s_0}$ and $P_{G\mid S=s_1}$, across all pairs $s_0, s_1 \in \secretspace$;
as observed above, $P_{G\mid S}$ is a mixture of Gaussians.
Unfortunately, there are no known tight bounds on the divergence between mixtures of Gaussians.

This is not an unknown obstacle: all previous DP-SGD analyses (\eg, DP-based) have encountered its analog.
FFT-based accountant methods address this issue by discretization: for a fine enough grid, one can empirically measure the divergence between the distributions.
Recent work by \citet{mahloujifar2022optimal} uses Monte Carlo estimations, by sampling from the mixture distribution.
Both approaches, although valid, come with high computational costs.

In this paper, we study the benefits of a different strategy: we observe that the mixture distribution generated by DP-SGD
can be approximated with a Gaussian distribution \textit{for certain choices of parameters}.
Fortunately, these parameter choices happen to be of interest for most practical purposes.

\paragraph*{Approximating a mixture with a Gaussian.}
\begin{figure}
	\centering
	\includegraphics[width=0.8\linewidth]{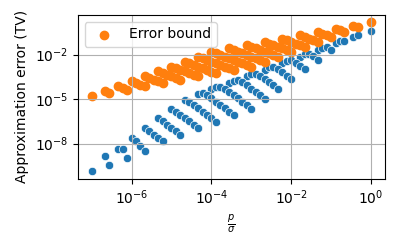}
	\caption{
		We compare the error induced by approximating a mixture of Gaussians with a Gaussian (\Cref{thm:small-approximation-error}). The error, measured as the total variation between the original and approximate distributions, is computed via numerical integration
		for a fixed $T=1$.
		A small ratio between the sampling rate $\samplerate$ and the noise parameter $\sigma$ ensure the error is negligible.
    }
	\label{fig:tv-error-vs-ratio}
\end{figure}

The proof of our main result relies on computing the total variation between two Gaussian mixtures.
Our first observation is that, in some cases, a mixture of Gaussians can be approximated by a Gaussian.
We formalize this in the following result, which shows the error committed when making this approximation in terms of the total variation between the original and approximate distributions.
For clarity, we let $\samplerate = \nicefrac{L}{N}$.

\begin{restatable}{proposition}{thmsmallapproximationerror}
\label{thm:small-approximation-error}
Let $f_{\mathcal{M}}$ be a Gaussian mixture defined as follows.
For a mean vector $\mu = (\mu_1, ..., \mu_T)$ and covariance matrix $\sigma^2 C^2 \Id_T$,
and $C = \max_{j=1}^T \mu_j$, let
$f_{\mathcal{M}}(x) = \sum_{b \in \{0, 1\}^T} \pi_b f_{\mathcal{N}(\mu b, \sigma^2C^2)}(x) \,.$
The $i$-th component takes values from $f_{\mathcal{N}(\mu_i, \sigma^2C^2)}$ with probability $\samplerate \in [0,1]$, or from $f_{\mathcal{N}(0, \sigma^2C^2)}$ otherwise.
Here, $\pi_b = \samplerate^{|b|}(1-\samplerate)^{T-|b|}$.
The error committed in approximating $f_{\mathcal{M}}$ with
$f_{\mathcal{N}(\samplerate\mu, \sigma^2C^2)}$ is:
$$\TotalVar(f_{\mathcal{M}}, f_{\mathcal{N}(\samplerate\mu, \sigma^2C^2)})
	= O\left(\frac{\sqrt{\samplerate T}}{\sigma}\right)$$
\end{restatable}
Proofs are in the appendix.

\paragraph{The total variation between two Gaussians.}
The Bayes security of DP-SGD reduces to computing the total variation between two Gaussian distributions that are identically-scaled (with isotropic covariance matrix).
For this step, we use the following closed form expression, which was derived by \citet{Devroye:2018} using a result by \citet{Barsov:1987}:

\begin{corollary}[From \citet{Barsov:1987}, Theorem 1]
\label{thm:tv-bound}
Let $\mu_0, \mu_1 \in \mathbb{R}^{m \times n}$, $\sigma > 0$. Then,
\begin{equation*}
\TotalVar\left(\N{\mu_0}{\sigma^2 \mathbf{I}}, \N{\mu_1}{\sigma^2 \mathbf{I}}\right)
= \erf\left(\frac{\fnorm{\mu_0 - \mu_1}}{2\sqrt{2}\sigma}\right)
\end{equation*}
where $\fnorm{A} = \sqrt{\mathrm{tr}(A A^\intercal)}$ is the Frobenius norm.
\end{corollary}

\paragraph{Main result.}

We can now state our main result: a closed-form bound on the security of DP-SGD against record-level property inference.
The bound depends on a variable, $\Delta_\prop$, whose value depends on the threat model (and, consequently, property of interest $\prop$), defined as:
\begin{equation}
	\label{eq:delta-prop}
	\Delta_\prop = \max_{s_0, s_1 \in \codom(\prop)} \fnorm{\bar{\g}(\prop^{-1}(s_0)) - \bar{\g}(\prop^{-1}(s_1))} \,.
\end{equation}
Here, $\bar{\g}(\challenge) = (\bar{\g}_1(\challenge), ..., \bar{\g}_T(\challenge))$ is the sequence of gradients computed by DP-SGD on a challenge point $\challenge$;
$\prop^{-1}(s)$ is the challenge point that satisfies $\prop(\prop^{-1}(s)) = s$, which is unique by assumption.

Intuitively, $\Delta_\prop$ indicates how much influence each property value has on the gradients, and it takes higher values the more the gradients change when the property value changes; in particular, it captures the worst-case scenario, when the attacker has to distinguish between the two property values that leak the most information.
We will provide an explicit value for $\Delta_\prop$ for the case of MIA (\Cref{thm:mia}) and AI (\Cref{thm:ai}) in the next sections.

The Bayes security of DP-SGD against record-level property inference is as follows:

\begin{restatable}{theorem}{thmmainresult}
	\label{thm:main}
	Assume that $\prop$ is a bijection, and let $\Delta_\prop$ be defined as in \Cref{eq:delta-prop}.
	The Bayes security of DP-SGD with respect to the record-level property inference threat described in \Cref{sec:main-result} is:
	$$\beta^*(P_{O \mid S}) \geq 1-\erf\left(\samplerate\frac{\Delta_\prop}{2\sqrt{2}\sigma C}\right) - O\left(\frac{\sqrt{\samplerate T}}{\sigma}\right)\,.$$
\end{restatable}

The proof combines the approximation of a mixture with a Gaussian (\Cref{thm:small-approximation-error}) with the bound on the total variation between two Gaussians (\Cref{thm:tv-bound}).

In the next two sections, we apply this result to bound the security against MIA (\Cref{sec:mia}) and AI (\Cref{sec:attribute-inference}).

\section{Membership Inference}
\label{sec:mia}
\Cref{thm:main} gives a bound for the Bayes security of DP-SGD against a record-level property inference attack.
In this section, we apply this result to derive a bound on MIA,
we study its tightness in comparison with existing estimates based on DP accountants, and we show how our bound relates to other metrics, such as the TPR$@$FPR of an optimal attacker.

The Bayes security of DP-SGD against MIA follows as a corollary to \Cref{thm:main}:

\begin{restatable}{corollary}{corollaryMIA}
	\label{thm:mia}
	The Bayes security of DP-SGD against record-level MIA (Game~\ref{game:mia-record-generalized}) is:
        $$\beta^* \geq 1 - \erf\left(\samplerate\frac{\sqrt{T}}{\sqrt{2}\sigma}\right) - O\left(\frac{\sqrt{\samplerate T}}{\sigma}\right) \,.$$
\end{restatable}

\paragraph{DP-SGD parameters selection.}

\Cref{thm:main} makes it possible to \textit{cheaply} decide on which parameters to select before running DP-SGD, given a desired level of MIA-resilience.
Suppose that an application requires $\beta^* \geq 0.98$;
assuming a uniform prior between members and non-members, this corresponds to at most a $51\%$ attack success probability.
Furthermore, suppose we wish to train for $T = 5k$ steps.
We can select the noise $\sigma$ and sampling rate $\samplerate$ based on the relation:
\begin{equation}
\label{eq:selecting-parameters}
\samplerate =  \frac{\erf^{-1}(1-\beta^*)\sqrt{2}}{\sqrt{T}}\sigma\,;
\end{equation}
in this example, $\samplerate \approx 0.00035\sigma$, which guarantees the desired level of protection.
\Cref{fig:mia-sigma-vs-samplerate} shows this in general, for $T=5k$.

\begin{figure}
	\centering
	\includegraphics[width=0.75\linewidth]{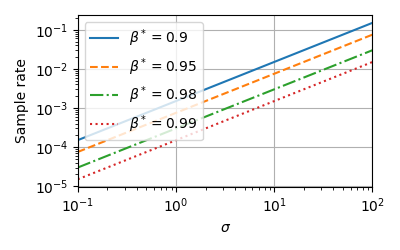}
	\caption{Bayes security against MIA: picking the noise and sampling rate to achieve a desired level of security ($T=5k$).
	}
	\label{fig:mia-sigma-vs-samplerate}
\end{figure}

\subsection{Comparison with the PLD accountant}
\label{sec:pld-comparison}

For the MIA threat model, there is a direct relation between Bayes security and $(0, \delta)$-DP:

\begin{proposition}[Bayes security and $(0, \delta)$-LDP~\citep{chatzikokolakis2020bayes}]
	Let $\mathcal{M}: \mathbb{S} \rightarrow \mathbb{O}$
	be a randomized algorithm that is also $(0, \delta)$-LDP,
	and assume $\mathbb{S} = \{0, 1\}$.
	Then:
	$$\beta^*(\mathcal{M}) = 1-\delta\, .$$
	\label{thm:relations-beta-dp}
\end{proposition}

Thanks to this relation, we can compare
our security bounds with equivalent ones estimated via state-of-the-art $(\varepsilon, \delta)$-DP numerical accountants.
We use the PLD accountant~\cite{fft1,fft2}, which supports the substitution adjacency relationship; this matches our MIA threat model (Game~\ref{game:mia-record-generalized}), where the attacker has to distinguish two datasets that differ on a single record.\footnote{The PLD accountant also supports the add-remove adjacency relationship, which is not relevant for our threat model.}

The goal of this comparison is twofold. First, we evaluate under what parameter choices our bounds are tight (i.e., when the approximation error derived in \Cref{thm:small-approximation-error} is small).
Second, we compare their computational costs.

\paragraph*{The tightness of our MIA bound.}
Our main result (\Cref{thm:main}) and, consequently, our bound on MIA (\Cref{thm:mia}) are based on the approximation of a mixture of Gaussians to a Gaussian distribution (\Cref{thm:small-approximation-error}).
Naturally, we do not expect this approximation to work well for all parameter choices.
Based on our initial evaluation for $T=1$ (\Cref{fig:tv-error-vs-ratio}), we suspect it will perform better for larger values of $\sigma$.

In these experiments, we compare $\beta^*$ obtained as in \Cref{thm:mia} with a $(0, \delta)$ estimate given by the PLD accountant; because PLD is tight up to discretization error, we use its estimate $\beta^* \approx 1-\delta$ as the ground truth for these experiments.
\Cref{fig:tightness} shows the absolute error between the two estimates.
We observe that our bound has a small error ($\leq 0.01$) for $\sigma \geq 1$ up to
50 epochs (i.e., $T=50k$ for $p=0.001$);
as the number of epochs grows to 100, the error increases to $\approx 0.02$.
This shows that our bound is tight for a wide range of realistic parameters.
On the other hand, we observe that for $\sigma < 1$ the error is large; intuitively this is because
the approximation of a mixture with a Gaussian gets worse the smaller the variance of the Gaussian is.
In practice, this means it is not advisable to use our bound for $\sigma < 1$.

\begin{figure}
    \centering
    \includegraphics[width=0.75\linewidth]{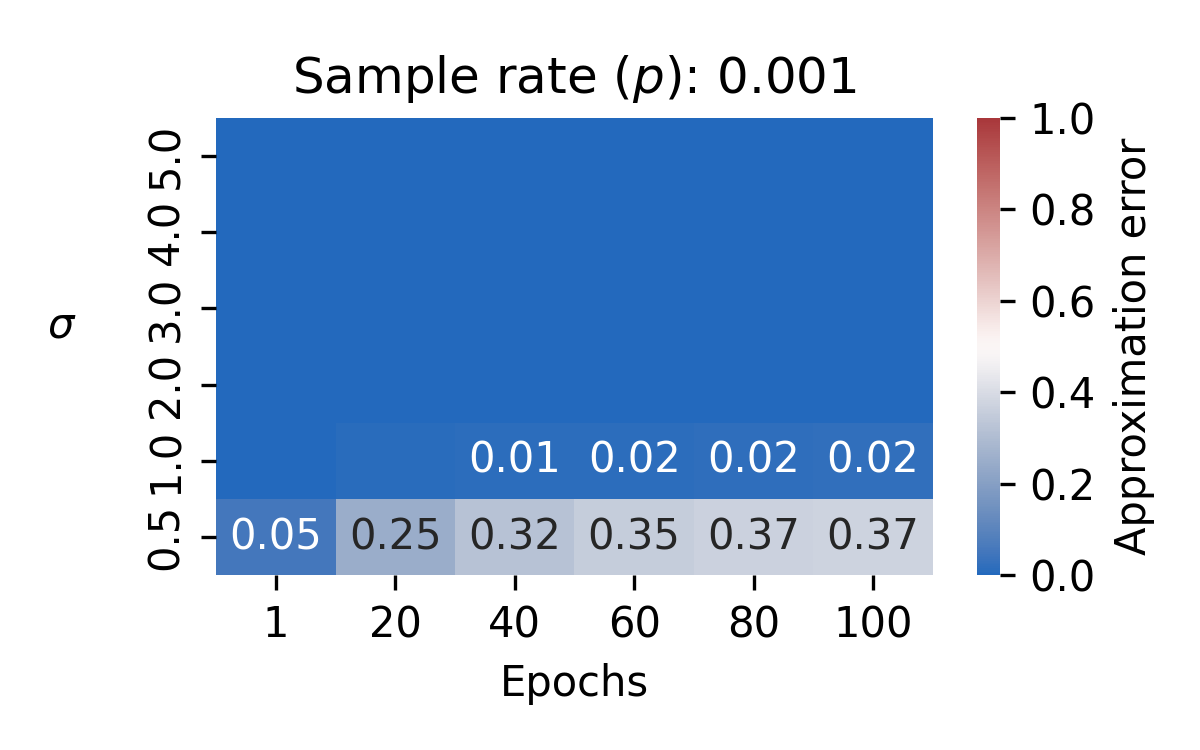}
    \caption{Approximation error between our bound (\Cref{thm:mia}) and the PLD accountant,
    w.r.t. the noise level $\sigma$ and the number of epochs ($\samplerate T$), with a sample rate $\samplerate = 0.001$.
    The error is measured as the absolute difference between the two estimates.
    A label for the error is shown only if it is $\geq 0.01$.}
    \label{fig:tightness}
\end{figure}

\paragraph*{Computational efficiency.}
We compare the costs of our bound with the PLD accountant.
\Cref{fig:time-comparison} shows that our bound is orders of magnitudes faster to compute than the respective PLD estimate.\footnote{Comparisons with other state-of-the-art accountants led to similar conclusions.}
The time efficiency of our bound enables practitioners to select the parameters for DP-SGD interactively, in real-time, and with a high level of accuracy.

\begin{figure}
    \centering
    \includegraphics[width=0.75\linewidth]{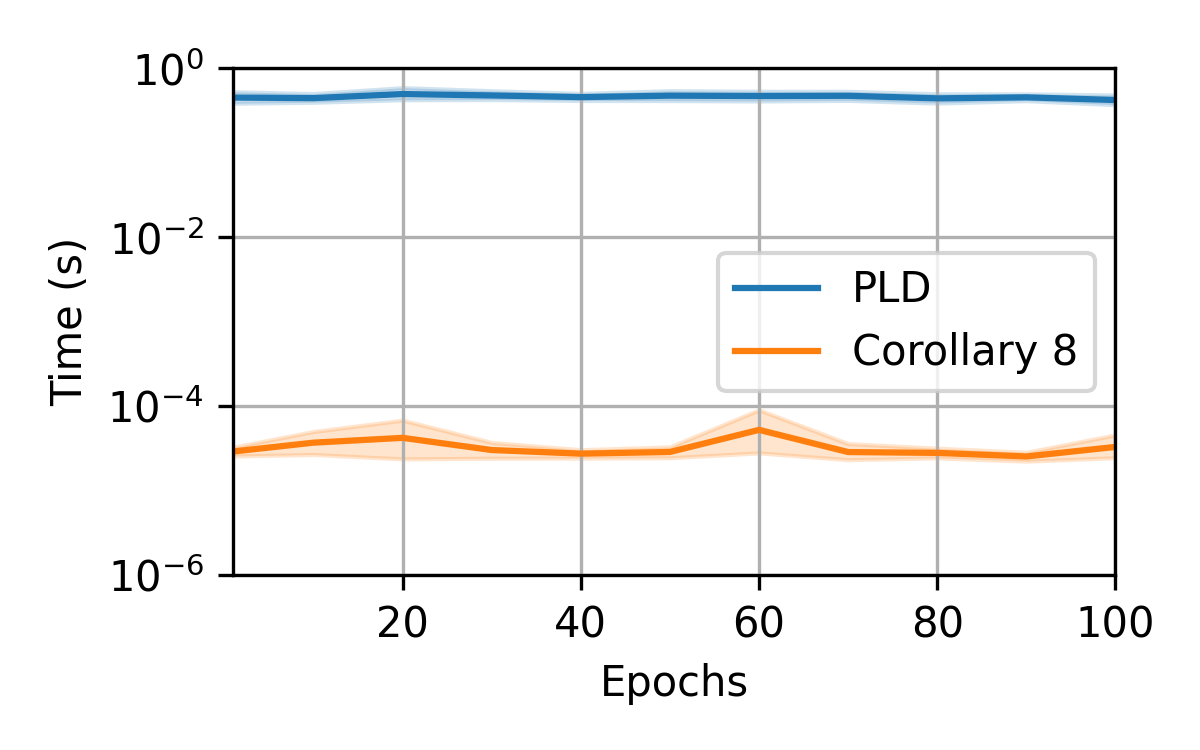}
    \caption{Computational cost of our bound (\Cref{thm:mia}) and the PLD accountant,
    w.r.t. the number of epochs, with a sample rate $\samplerate = 0.001$.
    The time is measured in seconds.}
    \label{fig:time-comparison}
\end{figure}

\subsection{Bayes Security and $(\varepsilon, \delta)$-DP}

An important drawback of $(\varepsilon, \delta)$-DP is that it may be harder to match to a specific threat model; in turn, this makes it difficult to select appropriate values $(\varepsilon, \delta)$.
In this section, we compare Bayes security with DP in terms of their ability to capture metrics of interest for MIA.
We do this analysis by relating Bayes security and $(\varepsilon, \delta)$-DP to two quantities that are commonly used for evaluating MIA threats: the advantage, and the true positive rate at a certain false positive rate (hereby denoted by TPR$@$FPR).

\paragraph*{MIA advantage.}
The MIA advantage, $\advantage$, measures how much more likely an attacker is to guess the membership of a data record having access to the trained model, compared to an attacker who only guesses based on prior knowledge.
Intuitively, $\advantage$ describes the additional risk that one incurs by releasing the model, w.r.t. the MIA threat.

The equivalence between Bayes security is  direct: $\beta^* = 1-\advantage$.
The relation between $\advantage$ and $(\varepsilon, \delta)$-DP was shown by \citet{Humphries}: if a mechanism is $(\varepsilon, \delta)$-DP, then the advantage is bounded as follows:
\begin{equation}
    \label{eq:humphries}
    \advantage \leq \frac{e^\varepsilon-1+2\delta}{e^\varepsilon+1} \,.
\end{equation}

In \Cref{fig:eps-dp-comparison}, we illustrate the behavior of this bound for progressively decreasing values of $\varepsilon$.
The curves are obtained by computing $(\varepsilon, \delta)$ via the PLD accountant for DP-SGD, and then plugging them into \Cref{eq:humphries}.

Results indicate that, by taking smaller values of $\varepsilon$, we get tighter bounds on the advantage.
In particular, the tightest bound is achieved when $\varepsilon = 0$; this corresponds to the case when $\beta^* = 1-\delta$.
This validates the use of a notion related to $(0, \delta)$-DP: under this configuration, we can hope to achieve the tightest analysis from an advantage perspective.

\begin{figure}
	\centering
	\includegraphics[width=0.75\linewidth]{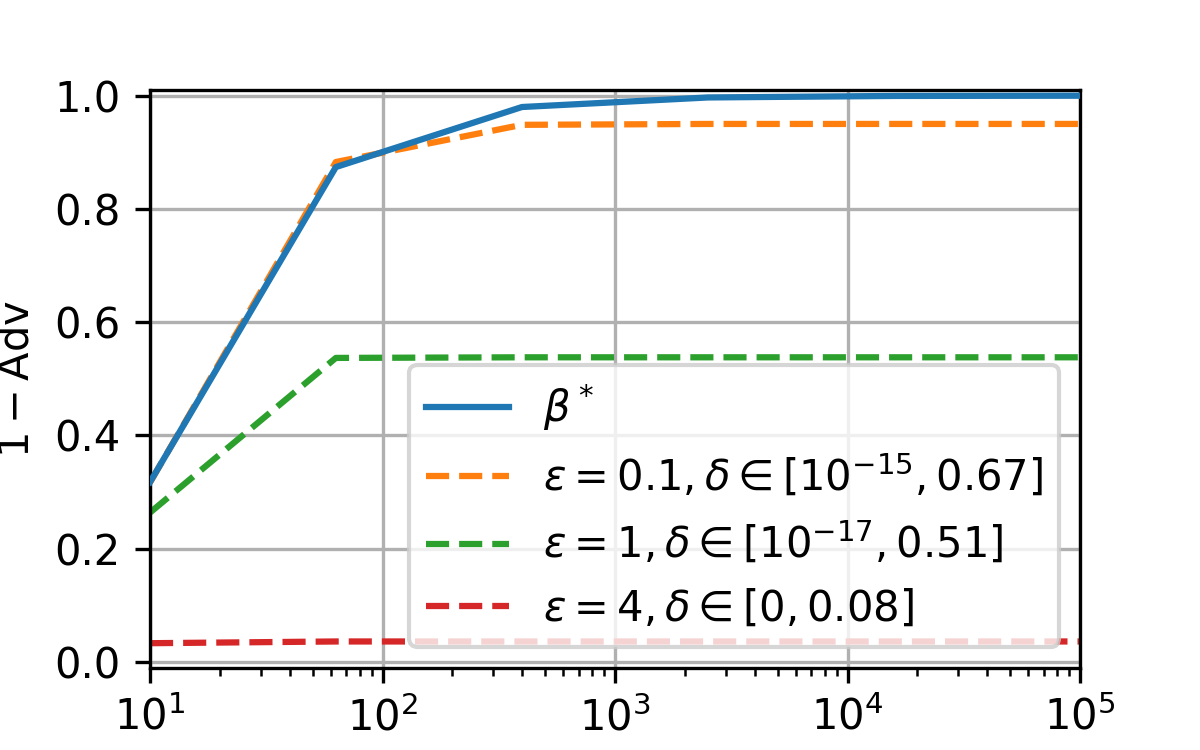}
	\caption{$(\varepsilon, \delta)$-DP of DP-SGD.
	We set $N=100k$, $L=10$, $C=1$, $T=1$. Here, $\beta^*$ is computed via \Cref{thm:mia}.}
	\label{fig:eps-dp-comparison}
\end{figure}

\paragraph*{TPR$@$FPR.}

\citet{carlini2022membership}
recommended
measuring the true positive rate of attacks at low false positive rates (\ie, TPR$@$FPR).
Their reasoning is that attacks with high accuracy may be unhelpful in practice; \eg, an attack can have $99\%$ accuracy yet not be able to identify members confidently without an unreasonable number of false positives.

To facilitate a comparison, we first prove a relation between Bayes security and TPR$@$FPR:
\begin{restatable}{proposition}{thmtprfpr}
    \label{prop:tpr-fpr}
    Consider a randomized mechanism $\mathcal{M}: \secretspace \rightarrow \observablespace$ with $\secretspace = \{0, 1\}$,
    and let $S$ be a random variable on $\secretspace$ with $\pi = \Pr[S=1]$.
    Let $s' = \adv(\pi, \mathcal{M}(S))$ be the guess that $\adv$ makes for $S$ given the output of the mechanism.
    Let the \textit{true positive} rate ($\tpr$) be the probability that attacker $\adv$ guesses correctly when $S=1$, and the \textit{false positive} rate ($\fpr$) be the probability that they guess incorrectly when $S=0$.
    If the mechanism is $\beta^*$-secure then for every attacker:
    \begin{align*}
        &\tpr \leq 1 + \fpr - \beta^* \quad \text{if } \pi \leq \nicefrac{1}{2}\\
        &\tpr \leq \frac{\pi}{1-\pi} \left(1 + \fpr - \beta^*\right) \quad \text{otherwise.}
    \end{align*}
    We have equality for a uniform prior, $\pi = \nicefrac{1}{2}$.
\end{restatable}

We observe that the special case $\pi = \nicefrac{1}{2}$ was known (e.g., \cite{Yeom:2020}).
However, we believe the general case is novel.
In practice, we expect the case $\pi \leq \nicefrac{1}{2}$ to be more relevant; for example, in MIA, the prior probability that a data record is a member is typically smaller than the alternative case.

We compare this with the bound given by $f$-DP, which gives the best possible bound on TPR$@$FPR;
we remark that obtaining $f$-DP bounds is computationally expensive.
In \Cref{fig:tpr-fpr-comparison}, each method gives an upper bound on the TPR for a chosen FPR value.
We observe that Bayes security is optimal for TPR$@$FPR=0.5; this matches the case when $\beta^*$ is the complement of the advantage.
The bound given by Bayes security becomes worse for lower levels of FPR. However, we observe a relatively small discrepancy with respect to the $f$-DP bound.
For example Bayes security bounds TPR$@$0.1FPR by $0.128$, while $f$-DP bounds it by $0.113$;
for a smaller FPR, Bayes security indicates TPR$@$0.01FPR $\leq 0.038$,
while the $f$-DP bound is $0.012$.

These experiments suggest that one can use the (cheap to compute) Bayes security to obtain a good bound on TPR$@$FPR, and then use $f$-DP to tighten the bound if needed.

\begin{figure}
    \centering
    \includegraphics[width=0.55\linewidth]{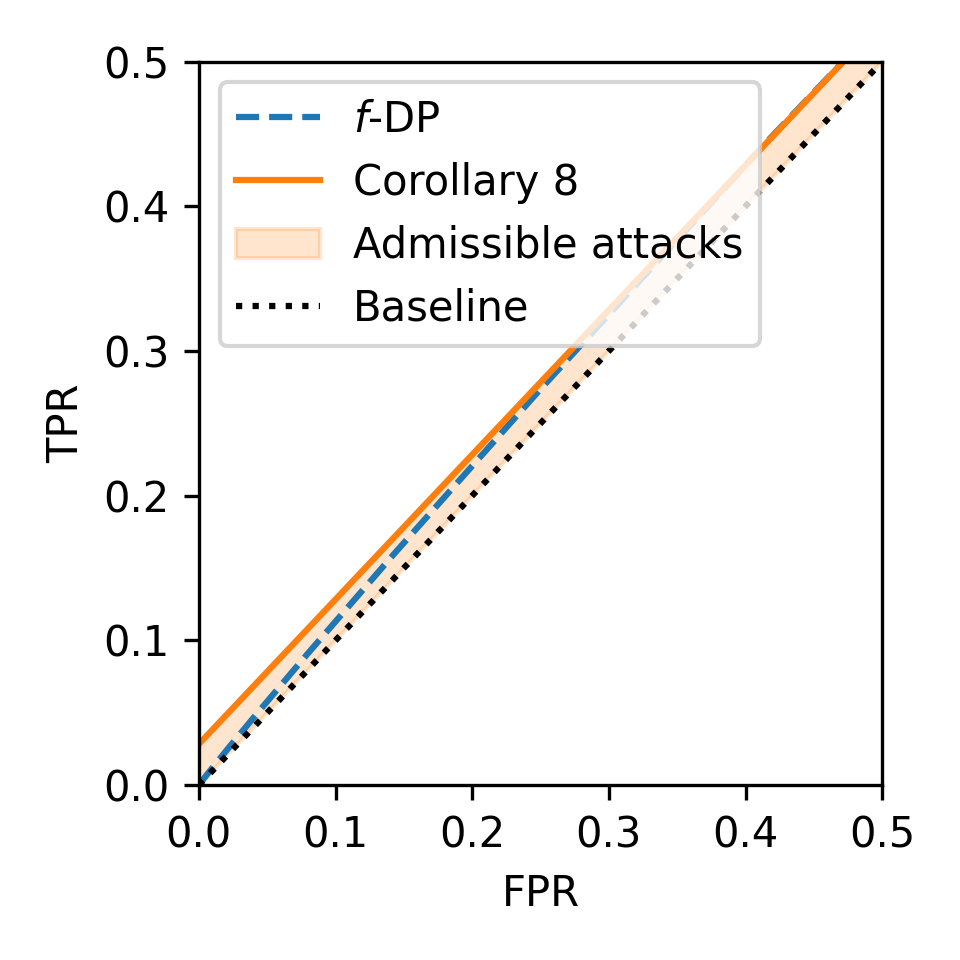}
    \caption{
        Bounds on TPR$@$FPR for DP-SGD, computed via Bayes security (\Cref{prop:tpr-fpr}) and $f$-DP.
        The viable region according to $\beta^*$ is highlighted in orange.
        DP-SGD parameters: $\samplerate = 0.0001$, $N=10k$, $50$ epochs (i.e., $T=\nicefrac{50}{p}$), $\sigma=2$.
    }
    \label{fig:tpr-fpr-comparison}
\end{figure}

\subsubsection{\Cref{thm:mia} as an $(\varepsilon, \delta)$-DP Estimator}
We observe that, in addition to having a direct expression for the security of DP-SGD against MIA, one could use \Cref{thm:mia} to obtain a rough estimate of $(\varepsilon, \delta)$-DP.
To this end, we can once again exploit \Cref{eq:humphries} by \citet{Humphries}.

By the correspondence between the advantage and Bayes security ($\advantage = 1-\beta^*$), we obtain a bound on $\varepsilon$ as follows.
Let $\beta^*$ be the Bayes security of DP-SGD computed as per \Cref{thm:mia};
then for any choice of $\delta \in [0, 1)$, we get:
$$\varepsilon \geq \log -\frac{2\delta + \beta^* - 2}{\beta^*} \,.$$

This bound can be a cheap alternative to more computationally expensive methods (\eg numerical accountants) for estimating $\varepsilon$.
However, it should be remarked that this bound is loose. The inequality by \citet{Humphries}, whilst tight, applies to \textit{any} $(\varepsilon, \delta)$-DP algorithm: one might improve on their inequality and the above estimate for the case of DP-SGD; this is what more advanced $(\varepsilon, \delta)$-DP estimators do.

\paragraph*{Related work.}
The $(\varepsilon, \delta)$-DP literature has explored the privacy guarantees of many basic mechanisms.
In particular, \citet{balle2018improving} and \citet{sommer2018privacy}
studied the privacy of a \textit{Gaussian mechanism without subsampling}.
By applying our observation that the distribution of the gradients can be approximated with a Gaussian, the special case $\samplerate=1$ of \Cref{thm:mia} can be obtained as a consequence of their results, thanks to the relationship between Bayes security and $(\varepsilon, \delta)$-DP.
Further, the effect of subsampling in the context of DP is understood to amplify the privacy level by a factor of $\samplerate$~\cite{Abadi:2016,wang2019subsampled,mironov2019r}.
In concurrent work, \citet{mahloujifar2022optimal} suggested the following strategy: for a specific threat (membership inference, in their case), determine the advantage of an attacker who observes the intermediate models output by DP-SGD.
Their proposal is to estimate this advantage via Monte Carlo simulations.
Our analysis strategy is similar to theirs in spirit: we aim to quantify the leakage for specific threats. Differently from them, we tackle a more general case (which subsumes membership and attribute inference), and we obtain closed-form expressions for our bounds.

\section{Attribute Inference}
\label{sec:attribute-inference}
In this section, we apply \Cref{thm:main}
to measure the security of DP-SGD against AI.
First, we discuss the limits of \textit{any} security analysis: without making assumptions, one cannot improve on MIA bounds (\eg, \Cref{thm:mia}).
We mitigate this issue by providing data-dependent bounds for AI, and by instrumenting DP-SGD to compute them.
Second, we study whether a data-dependent security analysis has any security implications.
Finally, we discuss the computational time overheads of our method and ways to improve it.

\subsection{Limits of \textit{any} AI Analysis}
Before stating our AI bound, it is important to understand what is achievable by a DP-SGD security analysis under this threat.
As it turns out, it is impossible to obtain a non-trivial bound for AI (\ie a bound for AI that is better than the MIA bound)
without making assumptions on the \textit{gradient function}.

This is easy to see by the following example.
In AI, the attacker tries to guess the secret value $S$ given partial information $\varphi$ and the model's weights.
By the arguments made in \Cref{sec:just-need-intermediate-gradients},
we can simplify this and limit the information available to an attacker to $\varphi$ and the clipped and noisy gradient: $\bar{\g}_t(\challenge) + \mathrm{Noise}$;
further, by the arguments made in \Cref{sec:mia-two-challenge-points}, we can assume there are just two secrets (\ie $S \in \{0, 1\}$).
Let us now consider a contrived gradient function, which returns $\bar{\g}_t(\concat{\varphi}{s_0}) = -\bar{\g}_t(\concat{\varphi}{s_1}) = (C, 0, ..., 0)$ for every $t$, where $C$ is the clipping gradient.
It is easy to see that this case matches record-level MIA, and that the bound on the Bayes security against AI will be \Cref{thm:mia}, which cannot be improved upon without further assumptions.
Even then, it is unclear what reasonable assumptions one could make on $\bar{\g}_t$ without affecting the validity of a security analysis.

We address this problem by instrumenting DP-SGD to keep track of the sensitivity $\|\bar{\g}_t(\concat{\varphi}{s_i}) - \bar{\g}_t(\concat{\varphi}{s_j})\|$,
for all training points $z$ and all possible attribute values $s_i, s_j$.
This comes with an extra computational cost, which is although acceptable for various real-world tasks.
In the next part of this section, we derive the bound on Bayes security, describe an algorithm for measuring the bound, and then describe optimization and approximation strategies.

\subsection{Bayes Security of DP-SGD against AI}

We can now adapt our main result (\Cref{thm:mia}) to the case of AI;
as before, we do this by specializing the definition of $\Delta_\prop$.
We write $\varphi(\challenge)$ to denote the non-sensitive part of $\challenge$.

\begin{restatable}{corollary}{corollaryAI}
	\label{thm:ai}
	The Bayes security of DP-SGD against AI is:
	$$\beta^*(P_{O \mid S}) \geq 1-\erf\left(\samplerate\frac{\|R\|}{2\sqrt{2}\sigma C}\right) - O\left(\frac{\sqrt{\samplerate T}}{\sigma}\right) \,,$$
	where $R = (R_1, ..., R_T)$ with
	$$R_t = \max_{\challenge \in L_t} \max_{s_0, s_1 \in \aiattributespace} \|\bar{\g}_t((\varphi(\challenge), s_0))-\bar{\g}_t((\varphi(\challenge), s_1))\| \,,$$
	where $L_t$ is the batch sampled at step $t$.
\end{restatable}

\begin{algorithm}[htb!]
	\caption{$\textrm{AI-resilient-SGD}(\{z_1,\ldots,z_{N}\}, \calL(\btheta)=\frac{1}{N}\sum_i \calL(\btheta, z_i), \eta_t, \sigma, L, C, \aiattributespace)$}
	\label{alg:ai-resilient-dpsgd}
	\DontPrintSemicolon
	{\bf Initialize} $\btheta_0$ randomly\;
	\For{$t \in [T]$}{
		Take a random sample $L_t$ with sampling probability $L/N$\;
		{\bf Compute gradient}\;
		{For each $i\in L_t$, compute $\g_t(z_i) \gets \nabla_{\btheta_t} \calL(\btheta_t, z_i)$}\;
		{\bf Compute gradient bound w.r.t. attribute's value}\;
		{\newcode{$R_t = \max_{\challenge \in L_t} \max_{s_0, s_1 \in \aiattributespace} \|\bar{\g}_t((\varphi(\challenge), s_0))-\bar{\g}_t((\varphi(\challenge), s_1))\|$}}\;
		{\bf Clip gradient}\;
		{$\bar{\g}_t(z_i) \gets \g_t(z_i) / \max\big(1, \frac{\|\g_t(z_i)\|_2}{C}\big)$}\;
		{\bf Add noise}\;
		{$\tilde{\g}_t \gets \frac{1}{L}\left( \sum_i \bar{\g}_t(z_i) + \mathcal{N}(0, \sigma^2 C^2 \Id)\right)$}\;
		{\bf Descent}\;
		{ $\btheta_{t+1} \gets \btheta_{t} - \eta_t \tilde{\g}_t$}\;
	}
	{\bf Output} $\btheta_T$ %
\end{algorithm}

We describe how DP-SGD can be adapted to compute the values $R_t$.
We observe that bounds computed in this manner are \textit{data-dependent}: the value of $R_t$ at step $t$
depends on the model's parameters at that step, and on the data itself.
In the next part of this section, we discuss why this has no privacy implications for the attack under consideration (AI).

\Cref{alg:ai-resilient-dpsgd} modifies DP-SGD for calculating $R_t$.
For every batch $L_t$ and every point $\challenge \in L_t$,
the algorithm augments $\varphi(\challenge)$ with all completions $s \in \aiattributespace$, and determines the maximum distance between the two.
The Bayes security is determined by plugging the vector $R = (R_1, ..., R_T)$ in \Cref{thm:ai}.

\subsection{Privacy Implications of Data-dependence}
One may wonder whether computing a security metric that depends on the data may have any privacy implications;
indeed, $\beta^*$, computed as per \Cref{alg:ai-resilient-dpsgd}, contains information about the secret.
The main concern raises when revealing $\beta^*$ to a malicious party: would they be
able to infer any privacy information about the training set if given access to it?

We analyze this concern w.r.t. to two threat models: MIA and AI.
We can describe each case similarly to Games~\ref{game:mia-record-generalized}-\ref{game:ai},
with the difference that the output $O$ communicated to the attacker is the security metric $\beta^*$.
Equipped with this information, the attacker's goal is to guess the secret $S$ (i.e., \textit{membership} of a challenge point
or \textit{attribute} value).

\paragraph{MIA.}
Revealing the Bayes security $\beta^*$ of DP-SGD against AI, computed as described in \Cref{alg:ai-resilient-dpsgd},
\textit{may leak} the membership of a data record.
An example that supports this claim follows.
Suppose the challenge points, $\challengezero$ and $\challengeone$, are both such that
$R_t = \max_{s_0, s_1 \in \aiattributespace} \|\bar{\g}_t((\varphi(\challengeb), s_0))-\bar{\g}_t((\varphi(\challengeb), s_1))\|$;
the parameter $R_t^b$, computed for either challenge point $\challengeb$, is maximized by that challenge point.
Further, suppose $R_t^0 \neq R_t^1$.
In this special case, the attacker can infer the challenge point from $\beta^*$.
Based on this observation, we recommend that whenever both MIA and AI are a concern, the security parameter estimated for the AI analysis is not revealed to the public.

\paragraph{AI.}
If the main concern is AI, no information is revealed by $\beta^*$ itself.
The reason for this is that $R_t$ is computed based on all possible values $\mathcal{A}$ for the attribute.
Therefore, $R_t$ (hence, $\beta^*$) is the same regardless of the value of the sensitive attribute $s$.
Therefore, if AI is the only threat of concern for a deployment, it is safe to reveal the estimated security metric to the public.

\subsection{Computational Costs and Optimization}

The cost of Algorithm~\ref{alg:ai-resilient-dpsgd} grows quadratically in the number of attributes $|\aiattributespace|$.
In our experiments (\Cref{sec:experiments}), we observe that the time overhead is acceptable for small attribute spaces.
Nevertheless, as $\aiattributespace$ grows, this cost becomes too high.
We explore two strategies for reducing this cost.

\paragraph*{Domain knowledge.}
As a first strategy, we can use the fact that Bayes security is maximized over two secret values only (\Cref{thm:mia-game-reduction}); in particular, these should be the two values $s_0, s_1 \in \aiattributespace$ that maximize the attacker's advantage.
In many practical applications, we can exploit domain knowledge to decide in advance what values will likely give the attacker the best advantage.
For example, consider the MNIST dataset, where each data record is a pixel matrix, and where each pixel is represented by a value in $[0, 1]$.
Suppose the sensitive attribute is one of such pixels.\footnote{We could equivalently define the risk for a set of pixels at once. A similar argument would apply.}
In this case, we can make the assumption that the two values maximizing the risk for the attribute will be the two extremes, $\{0, 1\}$.
This observation can substantially reduce the computation cost of Algorithm~\ref{alg:ai-resilient-dpsgd}.

\paragraph*{The point set diameter problem.}
A second strategy is to approximate the value $R_t$.
To this end, we observe that finding the distance between the two maximally distant gradients is an instance
of the well-known point set diameter problem, which is defined as follows.
Let $(M, d)$ be a metric space on a finite set $M$ for some metric $d$.
A solution to the point set diameter problem is an algorithm that returns $\mathrm{diam}_M = \max_{x, y \in M} d(x, y)$.
Various exact and approximate solutions exist for this problem~\cite{yao1982constructing,imanparast2019efficient,chatzikokolakis2020bayes}.
In this paper, we consider a simple $O(N)$ solution, where $N = |M|$, which gives a lower bound based on the triangle inequality:
for any choice $x \in M$, we have $\mathrm{diam}_M \leq 2\max_{y \in M} d(x, y)$.

Let $v$ be the mean vector of the gradients $\{\bar{\g}_t((\varphi(\challenge), a))\}_{a \in \mathcal{A}}$; we estimate $R_t$ as:
$$R_t \leq 2 \max_{a \in \mathcal{A}}\|\bar{\g}_t((\varphi(\challenge), a)) - v\| \,.$$
Naturally, the choice of a \textit{lower} bound here is security-motivated: it measures the worst-case for the victim.

Note that this estimate can be improved either by picking more carefully the point $v$, or by running this algorithm for various choices of $v$ and then choosing the one giving the tightest bound.
Despite the approximation given by the triangle inequality, in our experiments we
observed this approximation to be good enough.
Nevertheless, practical applications may consider solutions that give tighter bounds (\eg, \cite{imanparast2019efficient}).

\paragraph*{Related and Future work.}
With an appropriate choice of adjacency relationship, one can capture the AI threat in DP. One may wonder whether this observation enables adapting accountant-based analyses of DP-SGD to this threat.
We observe that DP-SGD uses gradient-clipping to bound sensitivity. Since gradients are computed per-example, there is no stronger data-independent sensitivity bound for AI than MIA when adapting neighboring datasets as differing in one attribute in one record~\cite{kifer2011no}. Unfortunately, attribute-DP mechanisms~\cite{zhang2022attribute} are impractical for ML.

Opportunities for future work include further improvements to the computational efficiency of our AI bounds analysis.
When using the approximate algorithm, the bottleneck of the analysis becomes computing the gradients of data records obtained by replacing their sensitive attribute.
A promising strategy is to use influence functions (IF)~\cite{cook1980characterizations,koh2017understanding} to approximate this more efficiently.
The main idea behind IF is to \emph{efficiently} approximate the addition and removal of a training point to a trained model via a Taylor approximation of a Newton step.
We observe that future work may explore further strategies.
In addition to using alternative solutions to the point set diameter problem, one could use optimization algorithms such as gradient descent to obtain the value of $R_t$ more quickly.
Future work may also explore approximations of the $R_t$ expression, \eg by using a Newton approximation by taking inspiration from the influence functions literature.

\section{Empirical evaluation}
\label{sec:experiments}
We evaluate our security analyses on models trained via DP-SGD on two datasets.
First, we study the computational costs of the AI analysis, and the effectiveness of its approximation (\Cref{sec:attribute-inference}).
Second, we compare the privacy-utility trade-offs that our MIA and AI analyses can offer.

\paragraph*{Datasets.}
We use two tabular datasets; this makes it meaningful to conduct an attribute inference analysis.
They are the Adult Census Income dataset (\adult) and the Purchase dataset (\purchase).
The \adult dataset has 32,561 records with 108 attributes each (after one-hot-encoding).
It contains data from the 1994 US Income Census, and the learning task is to predict the income of a person (precisely, whether it is above 50K/year or not), given attributes such as age and education.
Importantly, it has attributes taking more than 2 possible values; this facilitates a comparison between the ``full'' and ``approximate'' AI analyses.
We select \textit{age} to be the sensitive attribute for the AI analysis:
this attribute has 73 unique values, ranging between 17 and 90.
For the purpose of the AI analysis, we consider the entire range $\{17, 18, ..., 90\}$.

The \purchase dataset has 197,324 records and 600 attributes. Each record correspond to one customer, and each (binary) attribute indicates whether the customer bought a particular item.
This dataset enables evaluating how well our analyses scale to larger datasets.
For the AI analysis, we select the first attribute (purchase) to be the sensitive one.

\paragraph*{Models and setup.}
We train two fully connected neural networks as described by \citet{bao2022importance}, implemented via \texttt{pytorch}.
We instrumented \texttt{Opacus}~\cite{opacus} to support our AI analysis, as a callback function that is run at every step.

\paragraph{Privacy parameters.}
We use \Cref{eq:selecting-parameters} for selecting the privacy parameters $\samplerate$ and $\sigma$.
For illustration purposes, we aim at a MIA Bayes security $\beta^* = 0.9$ after $20$ epochs; this sub-optimal security against MIA enables observing the benefits of the AI analysis.
We run DP-SGD for $30$ epochs; this enables observing the behavior of $\beta^*$ after the predicted number of epochs.
We let $L=256$ for the \adult dataset, and $512$ for the \purchase dataset.
This, paired with the training set size, enables determining the following noise parameters (\Cref{eq:selecting-parameters}): $\sigma = 3.51$ (\adult) and $\sigma = 1.8$ (\purchase).

\subsection{AI analysis}

\paragraph{Running Time.}
First, we evaluate the computational costs of our analysis.
We present the measurements for the \adult dataset, which is harder to tackle from an AI analysis perspective; indeed, for each step we need to compute $N \times |\aiattributespace|$ gradients, where $N$ is the size of the training set, and $\aiattributespace$ are the possible values for the sensitive attribute; further, for every step we need to solve the point set diameter problem for the set of generated gradients, which is expensive (\Cref{sec:attribute-inference}).

We train the model enabling one of the following analyses: DP accountant, MIA, AI (approximate), AI (full).
As a baseline, we include the training time for the same model without DP.
We run this for 10 epochs (26 steps per epoch).
\Cref{fig:dpsgd-time-comparison} shows the average time taken to train for one epoch.

The cost of training with DP is roughly twice the cost of training without.
We can see that running our MIA analysis gives a rather marginal improvement over a DP accountant; indeed, we expect that the computational advantages of our analysis (\Cref{sec:mia}) would prove more useful during parameter search than during training.
Unsurprisingly, we can see that both AI analyses have large costs relatively to the others.
In \Cref{sec:goodness-analyses}, we observe that this may be a fair price to pay given their advantages from a privacy-utility perspective.

Finally, the approximate analysis reduces the costs by one order of magnitude. In the next part, we measure how much this approximation affects the quality of the analysis.

\begin{figure}
	\centering
	\includegraphics[width=0.75\linewidth]{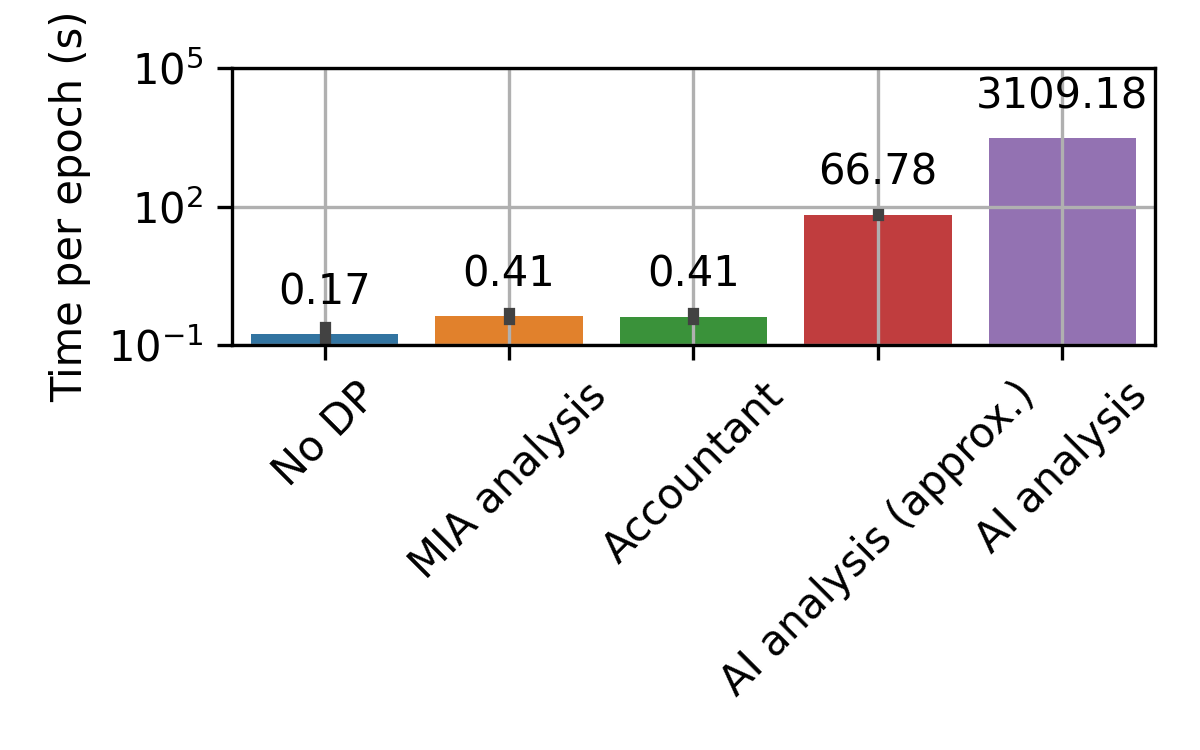}
	\caption{Average running time per epoch, across 10 epochs. \adult dataset (sensitive attribute has 73 possible values).
	}
	\label{fig:dpsgd-time-comparison}
\end{figure}

\paragraph{Full vs Approximate AI Analysis.}

In the previous part, we observed that the approximate AI analysis heavily reduces the computational costs.
\Cref{fig:ai-full-vs-approximate} compares the Bayes security estimates of the two analyses.
These experiments are run on the \adult dataset, whose sensitive attribute has more than 2 values; this enables appreciating the difference between the approximate and full version.
We remark that the approximate analysis can never indicate more security than the true one: it is a lower bound of the full AI analysis by construction.
The results suggest that the price one has to pay when running the approximate AI analysis as opposed to the full one is fairly small.
The estimated bounds on Bayes security after 20 epochs are $\beta^* \approx 0.945$ for the full analysis
and $\beta^* \approx 0.950$ for the approximate analysis.

Overall, while the costs of the approximate AI analysis are much larger than standard DP-SGD training, it is still a relatively scalable way of training, and it has advantages in the privacy-utility tradeoff, as illustrated in the next section.

\begin{figure}
	\centering
	\includegraphics[width=0.75\linewidth]{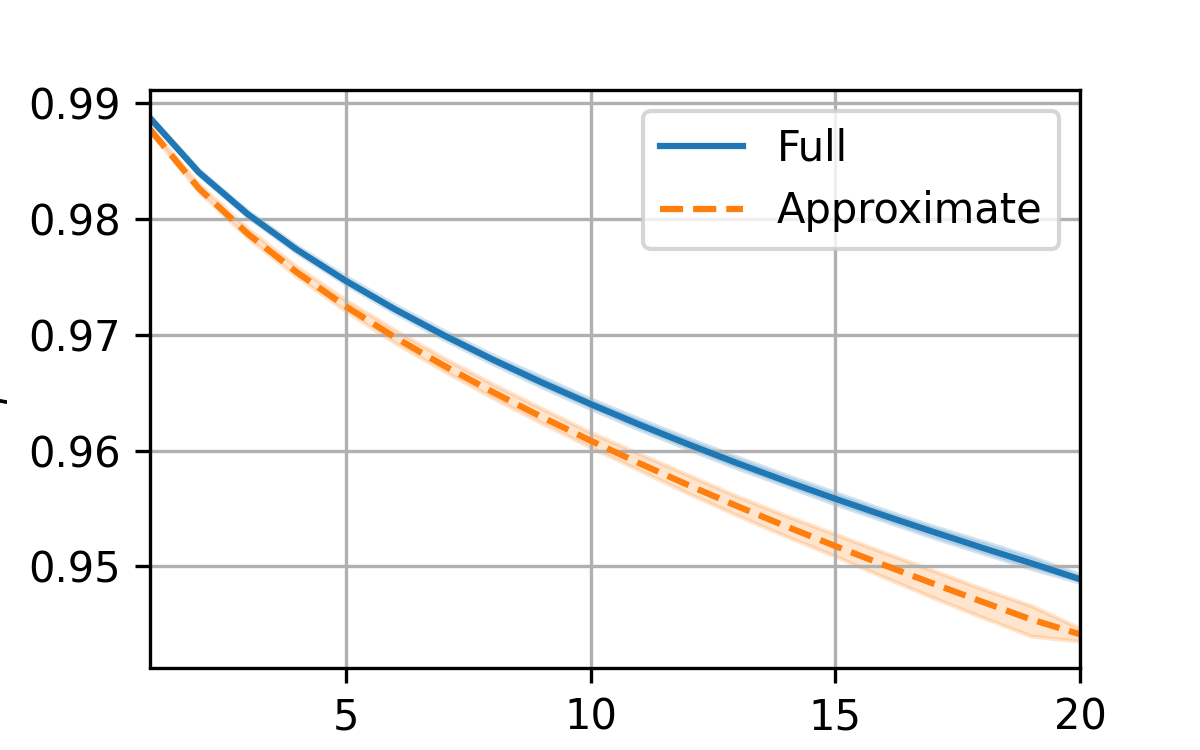}
	\caption{Bayes security on the \adult dataset, estimated via the full and approximate AI analysis.
	}
	\label{fig:ai-full-vs-approximate}
\end{figure}

\subsection{Bayes Security of DP-SGD}
\label{sec:goodness-analyses}

\Cref{fig:adult-purchase} relates the Bayes security (MIA and AI) with the accuracy of the trained model;
to avoid jitters, we round the accuracy to the closest multiple of 10\%, and show the corresponding confidence region.
For reference, we include the $(\varepsilon, \delta)$-DP estimate provided by a numerical accountant.

\paragraph*{Comparison with $(\varepsilon, \delta)$-DP.}
We observe that, for the chosen parameters, the risk of MIA is relatively high even for low values of $\varepsilon$:
for $\varepsilon \approx 0.5$, the risk of MIA is $\beta^* \approx 0.9$ (\adult).
By the relation between Bayes security and the attacker accuracy, and
assuming a uniform prior on the membership, this value of $\beta^*$ implies that the attacker can guess the membership correctly with 55\% probability.
Similarly, for an attacker aiming for a maximum FPR or 10\%, their TPR is at most 20\% (\adult).
This is a relatively high risk, considering that it matches a value of $\varepsilon$ that would generally be considered to be extremely safe.
Of course, this analysis depends on various assumptions: i) that the attacker knows the entire training set minus the challenge point, ii) that they can access all the gradient updates, and iii) that the MIA threat model is a concern for the particular deployment.
We now show empirically how relaxing the last assumption offers important privacy-utility benefits; we discuss the other two assumptions in \Cref{sec:white-box}.

\paragraph*{Empirical comparison between threat models.}
A principled way of relaxing a security analysis is to adopt a weaker threat model;
\eg, there are practical applications where MIA may not be a concern, but AI is.

Consider the best models trained on \adult (83\% accuracy) and \purchase (73\% accuracy).
We compare their resilience against MIA and AI by mapping their Bayes security to i) the attacker's success rate, and ii) TPR at FPR$\leq$10\%.
To compute the former, we need to assume a prior on the membership (resp., values of the sensitive attribute); for simplicity, we assume a uniform prior.
As for TPR$@$FPR, we notice it is not well-defined in general for AI.
We define it as follows: assume that there is a particularly damaging value of the sensitive attribute, which the attacker aims at predicting precisely; for example, in the \purchase dataset, the attacker may be interested in the value `1', corresponding to a product purchase.
We define TPR$@$FPR by considering this value to be the ``positives'' class.
The table below summarizes the results; we denote the attacker's success rate by $V$.

\begin{tabular}{lccc}
	& $\mathbf{\beta^*}$ & $\mathbf{V}$ & \textbf{TPR} \\
	\textbf{Task}& AI (MIA) & AI (MIA) & AI (MIA) \\
	\midrule
	\adult & 0.93 (0.88) & 53\% (56\%) & $\leq$ 16\% (20\%) \\
	\purchase & 0.99 (0.87) & 51\% (57\%) & $\leq$ 11\% (23\%) \\
	\bottomrule
\end{tabular}

We observe that, for both tasks, the AI analysis gives a much tighter bound on the attacker's probability of success, and on the TPR$@$FPR, than the MIA analysis.
In particular, we observe that best performing model on the \purchase dataset, while relatively unsecure against MIA (\eg, TPR$\leq$23\%), is almost perfectly secure against AI (TPR$\leq$11\%).
This emphasizes the benefits of a more nuanced analysis of the privacy risks of DP-SGD, which takes into account the attacker's knowledge and goals.

\section{Towards Inference-time Attackers}
\label{sec:white-box}

We apply our formalization to capture an \textit{inference-time} attacker, and discuss obstacles and future directions.

\paragraph{Training-time attacker.}
Both DP-guided analyses and our techniques have an important underlying assumption: that \emph{the attacker is able to inspect (and, possibly, modify) the intermediate gradients produced during training}.
This \emph{training-time} attacker is widely used throughout the literature, and has led to state-of-the-art results, such as the tight bounds for MIA obtained using the PLD accountant or the approach by~\citet{Humphries}.
However, this assumption is quite strong in practice.
Whilst it might hold in a federated learning setting, where the attacker can inspect (and possibly modify) the gradients during training, it it not representative of settings in which the adversary cannot observe the training process, \eg, a model trained privately in a secure environment.

\paragraph{Inference-time attacker.}
In contrast to the training-time attacker, the \emph{inference-time} attacker cannot observe the training process, and only has access to the final model.
Note that we can still consider the standard white-box vs.\ black-box duality for the inference-time attacker, which refers to whether the attacker has either full access to the model's weights or only the ability to query the model and receive responses.
Intuitively, the inference-time attacker has less information about the model than the training-time attacker, so it is reasonable to assume that the former should be weaker than the latter.
Evaluating security against an inference-time attacker may therefore yield significantly better bounds.

\paragraph{Modelling inference-time attackers with our framework.}
We expect our analysis techniques can be used for studying this weaker (albeit more realistic) threat scenario.
Formally, this is modelled by an attacker who tries to guess secret $S$ given only the final model weights $O_T$; note that the secret $S$ would be left unchanged.
As for deriving the bounds, one may be tempted to use a similar strategy to the one we used in \Cref{sec:main-result}.
Unfortunately, we found this to be non-trivial.

The main difficulty in applying our results to this problem comes from the effect that the challenge point $\challenge$, (possibly) sampled at step $t$, has on the subsequent gradient functions $\bar{\g}_{t+1}(\cdot), \bar{\g}_{t+2}(\cdot), \dots$.
In our analysis (\Cref{sec:just-need-intermediate-gradients}), we could disregard the effect of $\bar{\g}_{t}(z)$ for points $z \neq \challenge$ for $t > 1$.
The reason is that $\bar{\g}_{t}(z)$ did not bear \emph{additional} information about the challenge point to the adversary (who has access to $\bar{\g}_{t-1}(z)$ as well).
This reasoning does not apply to an inference-time attacker: since $\bar{\g}_{t}(\cdot)$ might have seen the effect of $\challenge$ in a previous step, and the attacker does not know $\bar{\g}_{t-1}(\cdot)$, we cannot disregard it as a potential source of leakage.

\paragraph{Explored directions.}
One way to approach this problem is via a taint analysis.
Let $\samplerate$ be the probability that the challenge point $\challenge$ is sampled at step $t$.
Then we can write a combinatorial expression that factors the likelihood that gradients (at points $z \neq \challenge$) were tainted by the presence of $\challenge$.
Unfortunately, this approach does not improve substantially on the analysis that we discussed in this paper:
since DP-SGD is usually run for a number of steps that is proportional to the number of batches, the probability ``$\challenge$ was sampled before step $t$'' grows exponentially with $t$; leading to trivial bounds.

We suspect that to obtain a tighter analysis of the security of DP-SGD against inference-time attackers, one may need to quantify the \emph{influence} of the challenge point on the gradient function.
Influence functions~\cite{hampel1974influence,koh2017understanding} may be a good tool to approximate this; however, it should be noted, this would require making assumptions about the gradient function.
An alternative is to use black-box leakage estimation methods (\eg, \cite{cherubin2019f}), which however require an infinite amount of data to provably convergence~\cite[Theorem 2.7]{cherubin2019black}.

\section{Conclusion}
\label{sec:conclusion}

The privacy of DP-SGD has historically been measured via DP-guided moment accountants.
This practice comes with various issues: i) computational complexity, ii) accountants are hard to implement correctly~\cite{fft2,doroshenko2022connect,gopi2021numerical}, iii) and DP is traditionally applied in a threat-agnostic manner.

Our proposal gives closed-form bounds for the privacy of DP-SGD, which are both easy to implement and orders of magnitude faster to compute than state-of-the-art DP estimators.
Additionally, our bounds are \textit{threat-specific}. This has two main benefits: i) they are arguably more interpretable, as they capture the risk for each threat individually; and ii) in circumstances where a weaker threat model (e.g., AI) is acceptable, one can achieve a much better utility at the same privacy level (\Cref{fig:adult-purchase}).
Finally, our bound on the resilience against the AI threat model is \textit{data-dependent}; this enables further pushing the utility depending on the inherent leakage of the data itself. Given the simplicity of modelling training algorithms (and respective threats) as information theoretic channels, we expect our analysis strategy can be used to derive bounds for other threats for existing training algorithms, or new ones designed with this strategy in mind.

\bibliographystyle{plainnat}
\bibliography{bibliography}
\balance

\appendix

\section{Proofs}

\thmsmallapproximationerror*

\begin{proof}
    First, observe that the total variation distance is related to the KL divergence $D_{KL}$ as follows:
    $\TotalVar(p_S, q_S) \leq \sqrt{\frac{1}{2} D_{KL}(p_S, q_S)}$.

    We use the following bound on the KL divergence between two Gaussian mixtures (see \citet{cover1999elements} and Eq. 13 in \cite{hershey2007approximating}):
    $$D_{KL}(f_{\mathcal{M}}, f_{\mathcal{N}(\samplerate\mu, \sigma^2C^2)}) \leq \sum_{b \in \{0, 1\}^T} \pi_b D_{KL}(f_{\mathcal{N}(\mu b, \sigma^2C^2)}, f_{\mathcal{N}(\samplerate\mu, \sigma^2C^2)}) \,,$$
    where the KL divergence between two $d$-variate Gaussians, respectively centered in $\mu_0$ and $\mu_1$, is:
    \begin{align*}
        D_{KL}(f_{\mathcal{N}(\mu_0, \sigma^2_0)}, f_{\mathcal{N}(\mu_1, \sigma^2_1)}) =& \frac{1}{2}\big(
        (\mu_0-\mu_1)^\intercal\Sigma^{-1}_1 (\mu_0-\mu_1)\\
        &+ \mathrm{tr}(\Sigma_1^{-1}\Sigma_0)
        - \ln \frac{|\Sigma_0|}{|\Sigma_1|} - T
    \big)\,;
    \end{align*}
    Observe that, for $\Sigma_0 = \Sigma_1 = \sigma^2\Id_T$, we have $\mathrm{tr}(\Sigma_1^{-1}\Sigma_0)
    - \ln \frac{|\Sigma_0|}{|\Sigma_1|} - T$ = 0.
    From the above, we obtain:

    \begin{align*}
        D_{KL}(f_{\mathcal{M}}, &f_{\mathcal{N}(\samplerate\mu, \sigma^2C^2)})
            \leq \sum_{b \in \{0, 1\}^T}  \frac{\pi_b}{2\sigma^2C^2}\left(
                \sum_{i=1}^T \mu_i^2(b_i-\samplerate)^2 \right)\\
            &\leq \sum_{b \in \{0, 1\}^T}  \frac{\pi_b}{2\sigma^2C^2}\left(
                \sum_{i=1}^T \max_{j=1}^T \mu_j^2(b_i-\samplerate)^2\right)\\
            &= \sum_{b \in \{0, 1\}^T}  \frac{\pi_b}{2\sigma^2}\left(
                \sum_{i=1}^T b_i^2 + \samplerate^2T - 2\samplerate \sum_{i=1}^T b_i
                \right)\\
            &= \sum_{b \in \{0, 1\}^T}  \frac{\pi_b}{2\sigma^2}\left(
                \samplerate^2T + (1-2\samplerate) \sum_{i=1}^T b_i
                \right)\\
            &= \frac{1}{2\sigma^2} \left(
                \samplerate^2T \sum_{b \in \{0, 1\}^T} \pi_b
                + (1-2\samplerate) \sum_{b \in \{0, 1\}^T} \pi_b |b|
            \right)\\
            &= \frac{1}{2\sigma^2} \left(
                \samplerate T - \samplerate^2 T
            \right)
    \end{align*}
    We used the fact that $b_i^2 = b_i$.
    \end{proof}

    For a fixed $T$, the goodness of this approximation depends on the choices of $\sigma$ and of the sampling rate $\samplerate$.
    In our main result, $T$ is the number of DP-SGD steps.
    The result matches expectations: if one wants to run DP-SGD for longer and retain strong security, one either needs to increase the noise multiplier or reduce the sampling rate.

    In \Cref{fig:tv-error-vs-ratio}, we compare the approximation error between the two distributions for a varying ratio between noise and sampling rate parameters.
    We observe that \Cref{thm:small-approximation-error} holds when the ratio is small.
    In particular, for realistic regimes with $\nicefrac{\samplerate}{\sigma} < 10^{-3}$, we observe a negligible approximation error ($<10^{-4}$).
    In \Cref{sec:mia}, we observe that these values for the parameters are not only practical; they are recommended to achieve stronger levels of security against MIA threats.

\thmmainresult*

\begin{proof}
By \Cref{thm:gradients-reduction}, the relation between Bayes security and the total variation \Cref{thm:relations-beta-tv}:
\begin{align*}
    \beta^*(P_{O \mid S}) &\geq \beta^*(P_{G \mid S})\\
        &= 1-\max_{s_0, s_1 \in \codom(\prop)} \TotalVar(P_{G \mid S=s_0}, P_{G \mid S=s_1})
\end{align*}

We now determine the total variation term.
Observe the following basic consequence of the triangle inequality.
Let $\nu_a$ and $\xi_a$ be two distributions parameterized by $a \in \{0, 1\}$.
Then:
$$\TotalVar(\nu_0, \nu_1) \leq \TotalVar(\xi_0, \xi_1) + \TotalVar(\nu_0, \xi_0) + \TotalVar(\nu_1, \xi_1) \,.$$
We use this to replace the Gaussians mixture $P_{G \mid S}$ with a Gaussian, replacing the two pairwise distances $\TotalVar(\nu_0, \xi_0), \TotalVar(\nu_1, \xi_1)$ with an error term as defined in \Cref{thm:small-approximation-error}.

For any two $s_0, s_1 \in \codom(\prop)$, we have:

\begin{align*}
	\TotalVar(&P_{G \mid S=s_0}, P_{G \mid S=s_1})\\
	=& \TotalVar(\sum_{b \in \{0,1\}^T} c_b P_{G \mid B=b, S=s_0}, \sum_{b \in \{0,1\}^T} c_b P_{G \mid B=b, S=s_1})\\
	\leq& \TotalVar(\mathcal{N}(\samplerate\bar{\g}(\prop^{-1}(s_0)), \sigma^2 C^2),
	\mathcal{N}(\samplerate\bar{\g}(\prop^{-1}(s_1)), \sigma^2 C^2))\\
	&+ O\left(\frac{\sqrt{\samplerate T}}{\sigma}\right)\\
	=& \erf\left(\samplerate\frac{\fnorm{\bar{\g}(\prop^{-1}(s_0))-\bar{\g}(\prop^{-1}(s_1))}}{2\sqrt{2}\sigma C}\right) + O\left(\frac{\sqrt{\samplerate T}}{\sigma}\right)
\end{align*}

In the first step, we used the above consequence of the triangle inequality.
In the second step, we used the fact that $P_{G \mid B=b, S}$ is a Gaussian distribution and applied \Cref{thm:tv-bound}.

\end{proof}

\thmtprfpr*

\begin{proof}
    \citet{chatzikokolakis2020bayes} show that Bayes security is the minimum of the ratio between the probability that the attacker guesses \textit{incorrectly} having access to the mechanism, $Pr[\adv(\pi, \mathcal{M}(S)) \neq S]$, and the probability that the attacker guesses incorrectly without access to the mechanism, $Pr[\adv(\pi) \neq S]$; that is:
    $$\beta^* \leq \frac{Pr[\adv(\pi, \mathcal{M}(S)) \neq S]}{Pr[\adv(\pi) \neq S]} \quad \forall \pi \in (0, 1)\,.$$

    We rewrite the above in terms of $\tpr$ and $\fpr$.
    Say the mechanism is run $k$ times, each time for a secret $S$ sampled according to the prior distribution $\pi$.
    Let $\fp$ and $\tp$ be the count of false positives and false negatives across these $k$ trials.
    Then:
    \begin{align*}
        &Pr[\adv(\pi, \mathcal{M}(S)) \neq S] = \frac{\fp + \fn}{k}\\
        &Pr[\adv(\pi) \neq S] = \frac{\min(\positives, \negatives)}{k}
    \end{align*}

    By combining the above, we have:
    \begin{align*}
        \beta^* &\leq \frac{Pr[\adv(\pi, \mathcal{M}(S)) \neq S]}{Pr[\adv(\pi) \neq S]}\\
            &= \frac{k}{\min(\positives, \negatives)}  \left(\frac{\fp + \fn}{k}\right)\\
            &= \frac{\positives}{\min(\positives, \negatives)} (\fpr + (1-\tpr))
    \end{align*}
    where $\positives = \pi k$ and $\negatives = (1-\pi)k$ are the number of positive and negative samples, respectively.
    The proof is concluded by considering separately the cases $\positives \leq \negatives$ and $\positives > \negatives$.
\end{proof}

We now discuss the proofs for \Cref{thm:mia} and \Cref{thm:ai}.

\corollaryMIA*

\begin{proof}
    Observe that in the MIA threat model (Game~\ref{game:mia-record-generalized}) $\prop(\challenge_s) = s$ is a bijection.
    Then:

    \begin{align*}
        \Delta_\prop
        &= \max_{\challengezero, \challengeone \in \trainset}\fnorm{\bar{\g}(\challengezero) - \bar{\g}(\challengeone)}\\
        &= \max_{\challengezero, \challengeone \in \trainset} \sqrt{\sum_{t=1}^T \|\bar{\g}_t(\challengezero) - \bar{\g}_t(\challengeone)\|^2}\\
        &\leq 2C\sqrt{T} \,.
    \end{align*}

    Applying \Cref{thm:main} concludes the proof.
\end{proof}

\corollaryAI*

\begin{proof}
	Observe that $\prop$ is a bijection: as per Game~\ref{game:ai}, for every $\challenge$, there is exactly one value $s \in \aiattributespace$ such that $\prop(\challenge) = s$.
	We bound $\Delta_\prop$ as defined in \Cref{eq:delta-prop}:

	\begin{align*}
		\Delta_\prop
			&\leq \max_{s_0, s_1 \in \aiattributespace, \challenge\in \trainset}\fnorm{ \bar{\g}(\concat{\varphi(\challenge)}{s_0})-\bar{\g}(\concat{\varphi(\challenge)}{s_1})} \\
            &=  \max_{s_0, s_1 \in \aiattributespace, \challenge\in \trainset} \sqrt{\sum_{t=1}^T\| \bar{\g}_t(\concat{\varphi(\challenge)}{s_0})-\bar{\g}_t(\concat{\varphi(\challenge)}{s_1}) \|^2} \\
			&\leq \sum_{t=1}^T \max_{s_0, s_1 \in \aiattributespace, \challenge\in \trainset} \| \bar{\g}_t(\concat{\varphi(\challenge)}{s_0})-\bar{\g}_t(\concat{\varphi(\challenge)}{s_1}) \|
	\end{align*}
	We then apply \Cref{thm:main}.
	Note that in this corollary's statement, we range $\challenge \in L_t$, where $L_t$ is the batch sampled at time $t$.
	This is allowed by observing that, if $\challenge$ is not included in the batch at step $t$, it cannot influence the model weights (and gradients) at that step.
\end{proof}

\section{Tightness of MIA bound (\Cref{thm:mia})}

We report results for further sample rates in \Cref{fig:tightness-appendix}.
Results for $\samplerate = 0.001$ (\Cref{fig:tightness}) are reported again, for comparison.
We observe that our observations hold for these parameters: the approximation error is small when $\sigma \geq 1$;
in general, a larger sampling rate $\samplerate$ may further increase this error, but not significatively.

We also note that the PLD accountant failed to provide a bound for $\samplerate = 0.0001$ and $\sigma=0.5$. Further, we observe that for the same sample rate, the approximation increases for a larger $\sigma$. We attribute this to numerical errors in the PLD accountant.

\begin{figure*}
    \begin{subfigure}[b]{0.5\linewidth}
        \centering
        \includegraphics[width=\linewidth]{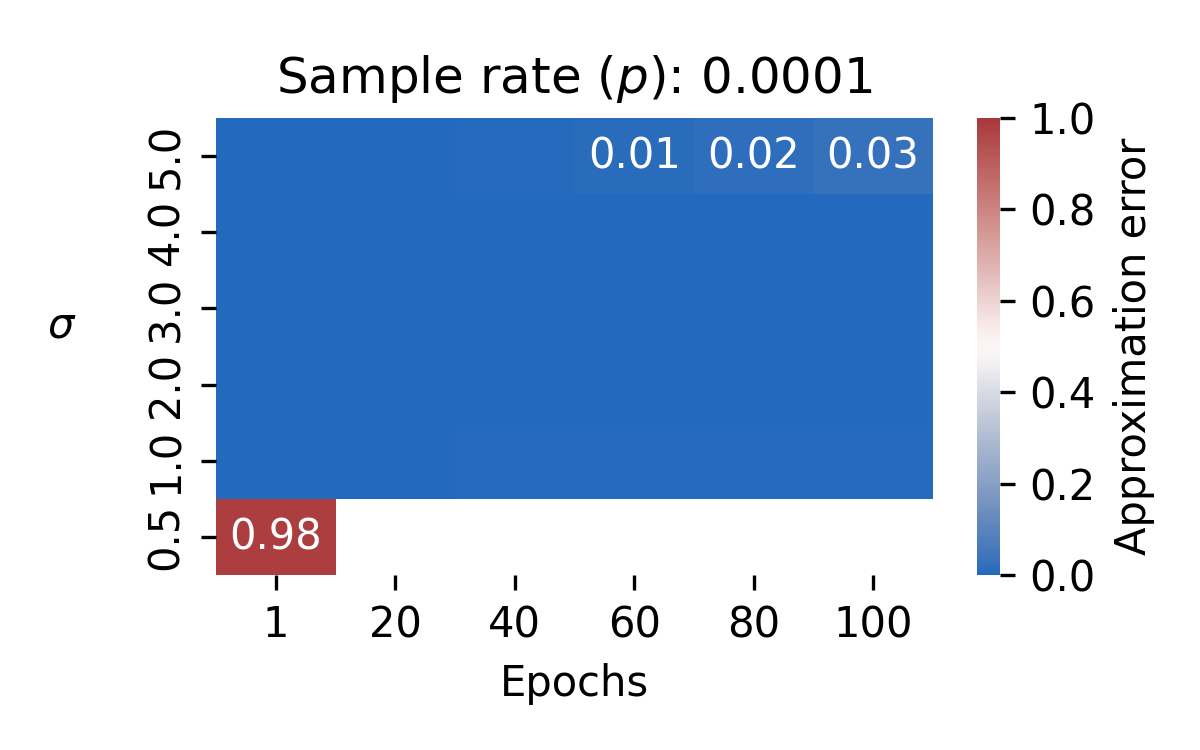}
    \end{subfigure}
    \begin{subfigure}[b]{0.5\linewidth}
        \centering
        \includegraphics[width=\linewidth]{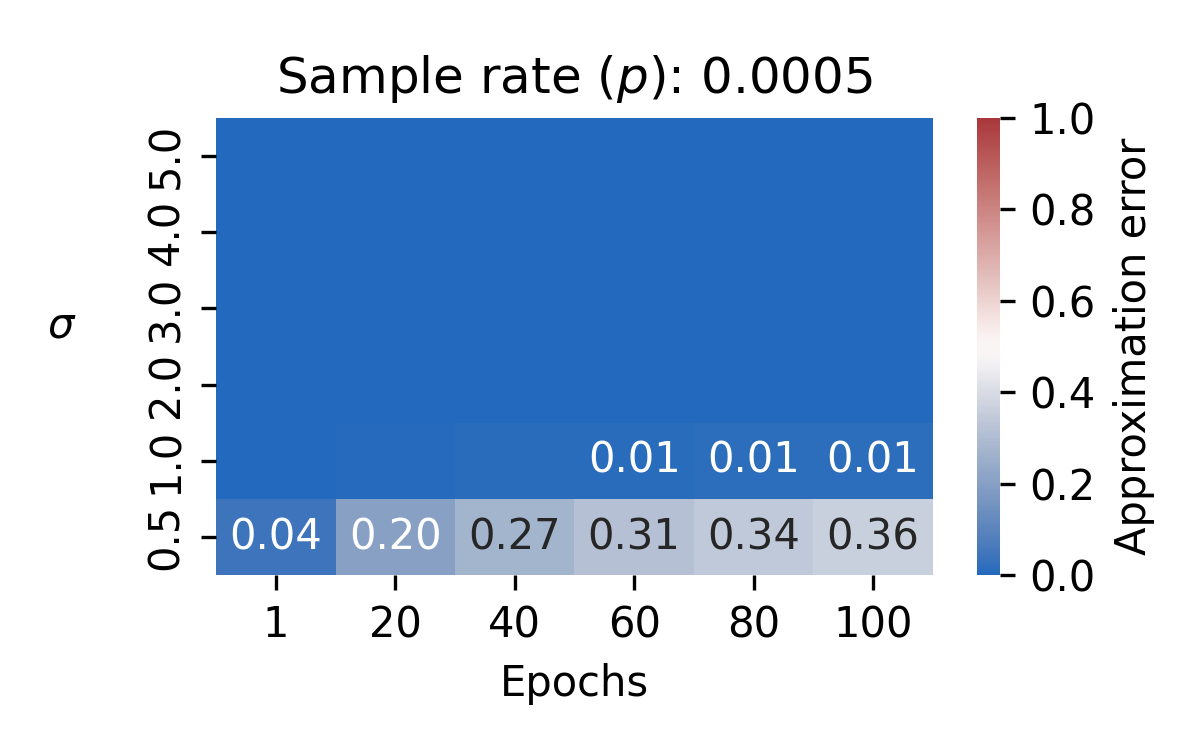}
    \end{subfigure}
    \begin{subfigure}[b]{0.5\linewidth}
        \centering
        \includegraphics[width=\linewidth]{img/accountant-approximation-error-p=0.001.png}
    \end{subfigure}
    \begin{subfigure}[b]{0.5\linewidth}
        \centering
        \includegraphics[width=\linewidth]{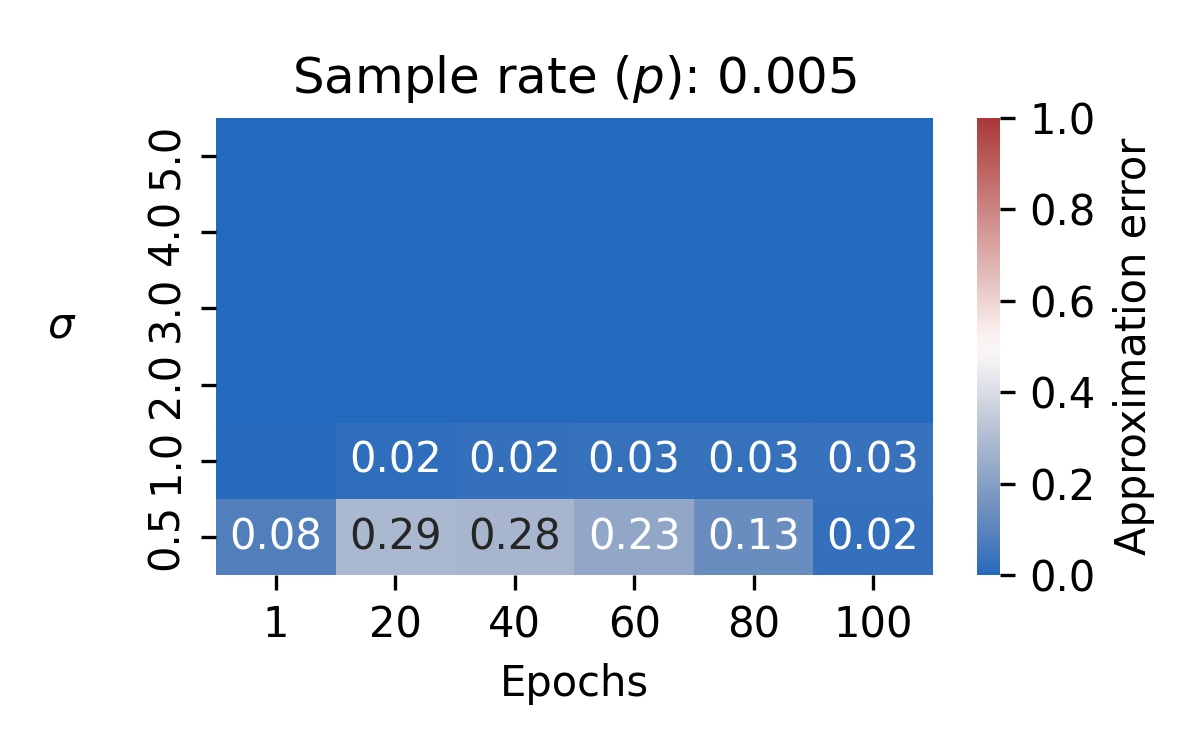}
    \end{subfigure}
    \begin{subfigure}[b]{0.5\linewidth}
        \centering
        \includegraphics[width=\linewidth]{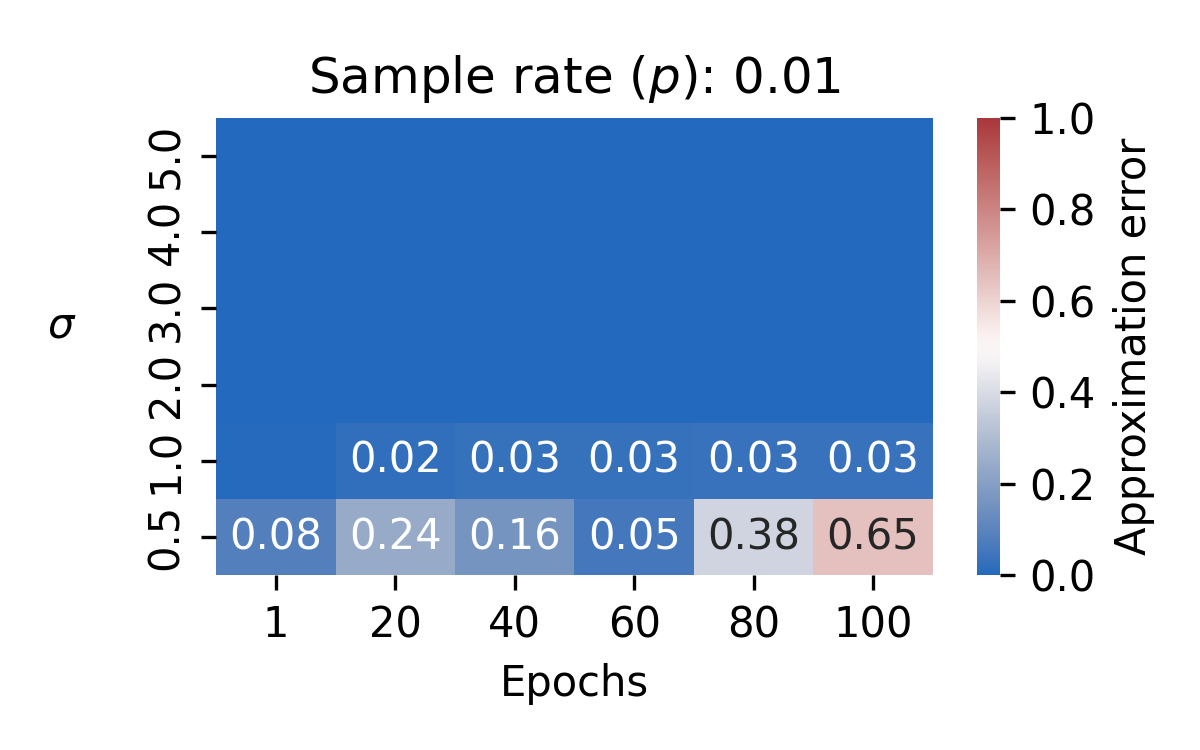}
    \end{subfigure}
    \caption{Approximation error of the PLD accountant for different sample rates $\samplerate$ and noise multipliers $\sigma$.}
    \label{fig:tightness-appendix}
\end{figure*}

\end{document}